\documentclass[a4paper,UKenglish]{article}

\usepackage{microtype}

\usepackage[right]{lineno}

\usepackage{	amssymb,
		amsmath,
		amsthm}

\usepackage{	xspace,
		enumerate,
		gensymb,
		xspace,
		graphicx,
		graphics,
		multirow,
		sidecap,
		enumerate}

\usepackage{	datetime}

\usepackage{	color, 
		colortbl, 
		hhline}

\usepackage{	arydshln}
\usepackage{	pst-node,
		pst-plot,
		pst-pdf}

\usepackage{	rotating}
\usepackage{	tikz}

\usepackage{hyperref}
\usepackage{ifmtarg}
\usepackage{chngcntr}
\usepackage{enumerate}

\clearpage{}

\newcommand{\mcomment}[2]{{\color{blue}(#1)}\footnote{#1: #2}}

\newcommand{\fhn}[1]{\mcomment{FH}{N}} \newcommand{\fhg}[1]{\mcomment{FH}{G}}

\newtheorem*{maintheorem}{Main Result}
\newtheorem*{maintechnicaltheorem}{Main Technical Result}

\newcommand{\fo}{\ensuremath{\mathrm{FO}}\xspace}
\newcommand{\fotwo}{\ensuremath{\mathrm{FO}^2}\xspace}

\newcommand{\emsotwo}{\ensuremath{\mathrm{EMSO}^2}\xspace}
\newcommand{\esotwo}{\ensuremath{\mathrm{ESO}^2}\xspace}

\newcommand {\rat}      {\mathbb{Q}}

\newcommand{\suc}{\ensuremath{{S}}}

\newcommand {\calA}      {{\cal A}\xspace}

\newcommand {\calB}      {{\cal B}\xspace}
\newcommand {\calC}      {{\cal C}\xspace}

\newcommand {\calK}      {{\cal K}\xspace}

\newcommand {\calP}      {{\cal P}\xspace}

\newcommand {\calL}      {{\cal L}\xspace}

\newcommand {\frakA} {\ensuremath{\mathfrak{A}}}

\newcommand {\frakB} {\ensuremath{\mathfrak{B}}}

\newcommand  {\N}   {\ensuremath{\mathbb{N}}}

\newcommand  {\myclass} [1]  {\ensuremath{\mbox{\sc #1}}\xspace}

\newcommand     {\NTWOEXPTIME}  {\myclass{2-NExpTime}}
\newcommand     {\TWONEXPTIME}  {\NTWOEXPTIME}

\newcommand     {\NEXPTIME}  {\myclass{NExpTime}}

\newcommand     {\EXPSPACE}  {\myclass{ExpSpace}}

\newcommand     {\FO}   {\myclass{FO}}

\newcommand{\df}{\ensuremath{\mathrel{\smash{\stackrel{\scriptscriptstyle{
    \text{def}}}{=}}}} \;}

 \clearpage{}
\clearpage{}

  \newtheorem{theorem}{Theorem}[section]
   \newtheorem{lemma}[theorem]{Lemma}
   
      \theoremstyle{definition}
    
    \newtheorem{example}[theorem]{Example}
    \newtheorem*{example*}{Example}

     \newenvironment{proofsketch}{\noindent\emph{Proof sketch.}\enspace}{\qed \\}

\newcommand{\reptheoremtitlefont}[1]{\textbf{#1}}
\newcommand{\reptheoremcontentfont}{\itshape}

\makeatletter
\newcommand{\theoremcont}[3]{   \def\Type{#1}   \def\Number{#2}   \def\Label{#3}  \@ifmtarg{#3}{     \reptheoremtitlefont{\Type\ \Number\ (R).}\reptheoremcontentfont \hspace{2mm}   }{    \reptheoremtitlefont{\Type\ \Number\ (R)}\ \reptheoremcontentfont(\Label).\hspace{-2mm}  }}

\clearpage{}
\clearpage{}
\pgfdeclarelayer{background}
\pgfdeclarelayer{foreground}
\pgfsetlayers{background,main,foreground}

\newcommand{\highlightedAreaFillColor}{blue!20}

\newcommand{\mnodeFillColor}{black!40}
\newcommand{\mnodeDrawColor}{black!60}

\tikzstyle{highlightedArea}=[
  fill=\highlightedAreaFillColor,
  rounded corners=2pt
]

\tikzstyle{mnode}=[
  circle,
  fill=\mnodeFillColor, 
  draw=\mnodeDrawColor,
  minimum size=4pt, 
  inner sep=0pt
]
 
\newcommand{\pictureTwoSuccessorsFarAwayA}{
  \begin{tikzpicture}[
    xscale=0.35,
    yscale=0.35
  ]
        \draw [->, line width=2pt, black!80] (0,0) -- (8,0) node[black, below] {$\suc_1$};
    \draw [->, line width=2pt, black!80] (0,0) -- (0,8) node[black, left, sloped] {$\suc_2, <_2$};

        \foreach \x in {1,2,...,7}
    \draw[dotted, black] (\x,0) -- (\x,8) ;

    \foreach \y in {1,2,...,7}
    \draw[dotted, black] (0,\y) -- (8,\y);

        \node (tmp) at (4, 4)[mnode, label={below left:$a$}] {};
    \node (tmp) at (6, 7)[mnode, label={above right:$b$}] {};

    \begin{pgfonlayer}{background}
      \fill [highlightedArea] (0,5.7) rectangle (2.3,8);
      \fill [highlightedArea] (5.7,5.7) rectangle (8,8);
    \end{pgfonlayer}
  \end{tikzpicture}
}

\newcommand{\pictureTwoSuccessorsFarAwayB}{
  \begin{tikzpicture}[
    xscale=0.35,
    yscale=0.35
  ]
        \draw [->, line width=2pt, black!80] (0,0) -- (8,0) node[black, below] {$\suc_1$};
    \draw [->, line width=2pt, black!80] (0,0) -- (0,8) node[black, left, sloped] {$\suc_2, <_2$};

        \foreach \x in {1,2,...,7}
    \draw[dotted, black] (\x,0) -- (\x,8) ;

    \foreach \y in {1,2,...,7}
    \draw[dotted, black] (0,\y) -- (8,\y);

        \node (tmp) at (4, 4)[mnode, label={below left:$a$}] {};
    \node (tmp) at (7, 1)[mnode, label={above right:$b$}] {};

    \begin{pgfonlayer}{background}
      \fill [highlightedArea] (0,0) rectangle (2.3,2.3);
      \fill [highlightedArea] (5.7,0) rectangle (8,2.3);
    \end{pgfonlayer}
  \end{tikzpicture}
}

\newcommand{\pictureExamplePointSetDataWord}{
  \begin{tikzpicture}[
    xscale=0.4,
    yscale=0.4
  ]
        \draw [->, line width=2pt, black!80] (0,0) -- (8,0) node[black, below] {$\suc_1, <_1$};
    \draw [->, line width=2pt, black!80] (0,0) -- (0,8) node[black, left, sloped] {$\suc_2, <_2$};

        \foreach \x in {1,2,...,7}
    \draw[dotted, black] (\x,0) -- (\x,8) ;

    \foreach \y in {1,2,...,7}
    \draw[dotted, black] (0,\y) -- (8,\y);

        \node (tmp) at (1, 5)[mnode, label={below:$\tau$}] {};
    \node (tmp) at (2, 6)[mnode, label={below:$\delta$}] {};
    \node (tmp) at (3, 3)[mnode, label={below:$\sigma$}] {};
    \node (tmp) at (4, 7)[mnode, label={below:$\delta$}] {};
    \node (tmp) at (5, 1)[mnode, label={below:$\tau$}] {};
    \node (tmp) at (6, 4)[mnode, label={below:$\tau$}] {};
    \node (tmp) at (7, 2)[mnode, label={below:$\delta$}] {};
    
   \end{tikzpicture}
}

\clearpage{}

\begin{document}

\title{Order-Invariance of Two-Variable Logic is Decidable}
\author{
Thomas Zeume\\
TU Dortmund University\\
  \texttt{thomas.zeume@cs.tu-dortmund.de}
\and 
Frederik Harwath \\
Goethe Universität Frankfurt am Main \\ 
\texttt{harwath@cs.uni-frankfurt.de}}

\maketitle

\begin{abstract}
It is shown that order-invariance of two-variable first-logic is decidable in the finite. This is an immediate
consequence of a decision procedure obtained for the finite satisfiability problem for existential second-order logic with two first-order variables ($\esotwo$) on structures with two linear orders and one induced successor. We also show that finite satisfiability is decidable on structures with two successors and one induced linear order. In both cases, so far only decidability for monadic $\esotwo$ has been known. In addition, the finite satisfiability problem for $\esotwo$ on structures with one linear order and its induced successor relation is shown to be decidable in non-deterministic exponential time.  

\end{abstract}

\section{Introduction}\label{section:introduction}

Order-invariance plays a crucial role in several areas of finite model theory. In descriptive complexity theory, for example, various well-known results establish that a logic captures a complexity class on structures that are equipped with a linear order. Usually, in such results, the particular order on a given structure is not important. That is, the formulas constructed in the proofs are \emph{order-invariant}, i.e. they do not distinguish different linear orders on a given structure (see Section \ref{section:invariance} for a precise definition).

First-order logic ($\fo$) is a logic of great importance in model theory and, consequently, order-invariant $\FO$-sentences have been studied in the literature before, see e.g. \cite{GroheSchwentick2000,BenediktSegoufin2009,EickmeyerElberfeldHarwath2014}, and the survey~\cite{Schweikardt13}. 

It is well-known that the question whether an $\FO$-sentence is order-invariant is undecidable. Two possible remedies are to restrict either the class of structures or the logic. For the former case, it is known that order-invariance remains undecidable for colored directed paths  \cite{BenediktSegoufin2009} and colored star graphs \cite{EickmeyerElberfeldHarwath2014}.

In this article we study the decidability of invariance for the \emph{two-variable fragment of first-order logic} (\emph{two-variable logic} or $\fotwo$ for short). This fragment is reasonably expressive and yet its satisfiability and finite satisfiability problems are decidable \cite{Scott1962,MortimerOn75}. As a query language, it has a strong connection to the XML query language XPATH on trees~\cite{Marx2005}. Furthermore, many modal logics can be translated to $\fotwo$ and inherit its good algorithmic properties. Those applications as well as the inability of $\fotwo$ to express transitivity of a relation has led to an exhaustive study of the complexity of the (finite) satisfiability problem of the logic where some relation symbols are interpreted by transitive relations, equivalence relations, linear orders, successor relations and preorders (see e.g. \cite{Otto01, KieronskiT09, Manuel10, BojanczykDMSS11, KieronskiO12, SchwentickZ12, ManuelZ13, SzwastT13, KieronskiMPT14}).

We establish the decidability of order-invariance for two-variable logic. 

\begin{maintheorem}
    Order-invariance of $\fotwo$ is decidable.
\end{maintheorem}
The key to our main result is a simple observation which relates the problem to a satisfiability problem on ordered structures. A $\fotwo$-sentence $\varphi$ is \emph{not} order-invariant if and only if there are a structure and two linear orders on its domain which are distinguished by~$\varphi$. More precisely, $\varphi$ is \emph{not} order-invariant if and only if there are two linear orders $<_{1}$ and $<_{2}$ on a finite set $D$ and a tuple of relations $\bar R$ on $D$ that interprets all relation symbols in~$\varphi$ except for $<$ such that $(D,\bar R,<_{1}) \models \varphi$ and $(D,\bar R, <_{2}) \not\models \varphi$.
The latter statement can be seen as a satisfiability question for a sentence of the \emph{two-variable fragment of existential second-order logic} ($\esotwo$) on finite ordered $(<_{1},<_{2})$-structures.
Here, a \mbox{$(<_{1},<_{2})$}-structure is \emph{ordered} if $<_{1}, <_{2}$ are interpreted by linear orders.
Thus, if the finite satisfiability problem for $\esotwo$ is decidable on ordered $(<_{1},<_{2})$-structures, then so is order-invariance.

We take a slightly more comprehensive approach and study the finite satisfiability problem for $\esotwo$ on structures that are equipped with linear orders and successor relations; henceforth called \emph{ordered structures}. 

Some prior work on decidability of $\esotwo$ and its monadic fragment $\emsotwo$ (where only quantification of unary relations is allowed) 
on ordered structures has been done under the guise of the finite satisfiability problem for $\fotwo$. An almost complete characterization of the classes of finite ordered structures for which $\emsotwo$ is decidable was obtained in a sequence of articles \cite{Manuel10, SchwentickZ10, Kieronski11, SchwentickZ12, ManuelZ13}. Only the case of finite ordered structures with at least three successor relations remains open.
For $\esotwo$ it is only known that it is decidable on the class 
of finite ordered $(<,\suc)$-structures \cite{GradelR99,Otto01, CharatonikW15} and $\NEXPTIME$-complete for $(\suc_1, \suc_2)$-structures \cite{CharatonikW13}. We note that the latter result combined with the observation from above establishes the decidability of successor-invariance. Here and in general, if both a symbol $<$ (or $<_i$) and $\suc$ (or~$\suc_i$) occur in the signature of a structure, we assume that $\suc$ (or~$\suc_i$) is interpreted by the induced successor of $<$ (or $<_i$).

\begin{figure*}[th!]
\begin{center}
 \renewcommand{\arraystretch}{1.1}

\scalebox{.8}{\begin{tabular}{|p{4.0cm}|p{5.5cm}|p{5.3cm}|}
\hline
 {\bf Class of ordered structures} &  $\emsotwo$ & $\esotwo$ \\
\hline \hline

\multicolumn{3}{|c|}{\textbf{One linear order/successor}}\\
\hline
 $\suc$-structures & \NEXPTIME-complete \cite{EtessamiVW02} & \NEXPTIME-complete $\bigstar$\\
 $<$-structures & \NEXPTIME-complete \cite{EtessamiVW02} & \NEXPTIME-complete \cite{Otto01} \\
 $(\suc, <)$-structures & \NEXPTIME-complete \cite{EtessamiVW02} & \NEXPTIME-complete $\bigstar$ \\

\hline
\multicolumn{3}{|c|}{\textbf{Two linear orders/successors}}\\
\hline
  $(\suc_1, \suc_2)$-structures & \NEXPTIME-complete \cite{CharatonikW13} & \NEXPTIME-complete \cite{CharatonikW13} \\
  $(<_1,<_2)$-structures & in \EXPSPACE \cite{SchwentickZ10} & in \TWONEXPTIME $\bigstar$ \\
  $(\suc_1, \suc_2, <_2)$-structures & \mbox{in  {\sc Multicounter-Emptiness}} \cite{ManuelZ13} & \mbox{in  {\sc Multicounter-Emptiness}} $\bigstar$\\
  $(\suc_1, <_1, <_2)$-structures &  in \EXPSPACE \cite{SchwentickZ12} &  in \TWONEXPTIME $\bigstar$\\
  $(\suc_1, <_1, \suc_2, <_2)$-structures & {Undecidable} \cite{Manuel10} & {Undecidable} \cite{Manuel10} \\

\hline
\multicolumn{3}{|c|}{\textbf{Many linear orders/successors}}\\
\hline

  $(<_1, <_2, <_3)$-structures & Undecidable \cite{Kieronski11} & Undecidable \cite{Kieronski11} \\
  $(\suc_1, \suc_2, \suc_3)$-structures & ? & ? \\
  $(\suc_1, \suc_2, \suc_3, \ldots)$-structures & ? & ? \\
\hline
\end{tabular}
}
\caption{Summary of results on satisfiability of $\esotwo$ and $\emsotwo$ on finite ordered structures. Results contained in this article are marked by $\bigstar$.} \label{fig:summary}
\vspace{0cm}
\end{center}
\end{figure*}
Our main technical contribution implies the decidability of order-invariance for $\fotwo$. 
\begin{maintechnicaltheorem}
  The satisfiability problem for $\esotwo$ is decidable on the class of finite ordered $(\suc_1,<_{1},<_{2})$-structures and on the class of finite ordered $(\suc_{1},\suc_{2},<_{2})$-structures.
\end{maintechnicaltheorem}

The first part  in particular closes the gap between the decidability
of $\esotwo$ on ordered $<$-structures \cite{Otto01} and the undecidability
of $\emsotwo$ on ordered $(<_{1},<_{2},<_{3})$-structures~\cite{Kieronski11}.
More precisely we show that the decision problem can be solved in $\TWONEXPTIME$ on $(\suc_1, <_{1},<_{2})$-structures and as fast as the emptiness problem for multicounter automata on $(\suc_{1},\suc_{2},<_{2})$-structures. We conjecture that there is an $\EXPSPACE$-algorithm for the former problem.
For proving our results we generalize techniques used in  \cite{Manuel10}, \cite{SchwentickZ10} and \cite{ManuelZ13}. In the course of this, a technique developed by Otto \cite{Otto01} for $\esotwo$ on ordered $<$-structures turns out to be very valuable also for the more general structures considered here. We emphasize that while the basic techniques used for proving the results in this article are inspired by previous work, they are employed in a technically more demanding context. 

As an introduction to the methods used for proving the technical main result, we show that $\esotwo$ is $\NEXPTIME$-complete on finite $(<, \suc)$-structures.
This is also of some independent interest. Previously, only the satisfiability problem for $\emsotwo$ was known to be $\NEXPTIME$-complete on this class of structures. The extension of $\esotwo$ with counting quantifiers was recently shown to be decidable on such structures,   but only with very high complexity~\cite{CharatonikW15}.

\subsection*{Discussion}
 
Our results on satisfiability of $\esotwo$ on finite ordered  structures also help to make known results more robust, and therefore less vulnerable for confusion. 

In work on two-variable first-logic originating from applications in verification and database theory often only unary relation symbols are allowed in formulas besides the linear orders and successors (see e.g. \cite{EtessamiVW02,BojanczykDMSS11}). In most articles this restriction is stated clearly, or, at least the intended meaning is clear from the context. Yet abbreviations like $\fotwo(<)$ do not reflect this restriction. The arising ambiguity can lead to confusions for non-experts and readers only skimming an article. We definitely have been confused a couple of times, and we seem not be alone.

The second-order perspective taken in this article resolves this confusion. Furthermore, our results obviate the necessity to distinguish the two variants in the context of decidability of the finite satisfiability problem to some extent.

\subsection*{Outline} 
We introduce our notation and basic tool set in Section~\ref{section:preliminaries}. 
In Section~\ref{section:warm-up} we present two different ideas for deciding $\esotwo$ on finite $(<,\suc)$-structures. These ideas are generalized to finite ordered $(\suc_1, <_{1},<_{2})$-structures and finite ordered $(\suc_1, \suc_{2},<_{2})$-structures in Sections~\ref{section:twoorders} and~\ref{section:twosuccessors}. We discuss order-invariance in Section~\ref{section:invariance}. We conclude in Section~\ref{section:conclusion}.

 \section{Preliminaries}\label{section:preliminaries}

In this section we introduce ordered structures and two-variable logic.

\subsection{Ordered Structures}
\label{sec:ordered-structures}
In this article we consider logical formulas that are interpreted over \emph{ordered} structures, i.e.~relational structures where some relation symbols are interpreted by linear orders and their corresponding successor relations. We assume that all structures are \emph{finite} structures with a \emph{non-empty domain}.

A {\em linear order} $<$ is a transitive, total and antisymmetric
relation, that is, $a  <  b$ and $b  <  c$ implies $a < c$ ; $a <  b $ or  $b  <  a$ holds; and not both $a  <  b$ and $b  <  a$ are satisfied at the same time  for all elements $a, b$ and $c$. The  \textit{induced successor relation} $\suc$ of a linear order $ < $ contains a tuple $(a, b)$ if $a <  b$ and there is no element $c$ such that $a <  c  <  b$.

  Let $O\subseteq \{<_1, <_2, \ldots, \} \cup \{\suc_1,\suc_2,\ldots, \}$ where the $<_{i}$ and $\suc_{i}$ are binary relation symbols.
  An $O$-structure $\frakA$ is \emph{ordered} if each relation  $<_{i}^{\frakA}$ is
  a linear order and each relation $\suc_{i}^{\frakA}$ is the induced successor relation of a linear order. Furthermore, if $\suc_i\in O$ and $<_{i} \in O$, then $\suc_{i}^{\frakA}$ is the induced successor relation of $<_{i}^{\frakA}$. For convenience we often identify relation symbols with their respective interpretations. We use the symbols $\suc$ and $<$ if an ordered structure only has one linear order. 
  
 We say that two distinct elements $a$ and $b$ of $\frakA$ are \emph{$<_i$-close} if $\suc_i(a,b)$ or $\suc_i(b,a)$. Otherwise, $a$ and $b$ are \emph{$<_i$-remote}. We say that $a$ and $b$ are \emph{close} if they are $<_i$-close for some $\suc_i \in O$. Otherwise, $a$ and $b$ are \emph{remote}. If the signature $O$ contains $\suc_i$ but not $<_i$ we also say $\suc_i$-close and so on.

 For an ordered structure $\frakA$ with two underlying linear orders $<_1$ and $<_2$ it will often be convenient to think of $\frakA$ as a point set in the two-dimensional plane where $<_1$ orders points along the $x$-axis and $<_2$ orders points along the $y$-axis (see Figure \ref{figure:example:PointSetDataWord}). Following this conception, we say that an element $a$ is to the \emph{left} or \emph{right} of an element $b$ if $a <_1 b$ or $b <_1 a$, respectively, and that it is \emph{below} or \emph{above} $b$ if $a <_2 b$ or $b <_2 a$.
 Accordingly, we will speak of \emph{leftmost} and \emph{rightmost} elements.

\subsection{Two-Variable Logic}
The \emph{two-variable fragment of first-order logic} (short: $\fotwo$) is the restriction of first-order logic where only two variables can be used, though those two variables can be quantified multiple times.
The \emph{two-variable fragment of existential second-order logic} (short: $\esotwo$)
consists of all formulas of the form $\exists \bar R \varphi$ where $\bar R$ is a tuple
of relation variables and $\varphi$ is a $\fotwo$-sentence.
Since each $\fotwo$-atom can contain at most two variables we assume, in the entire article and without loss of generality, that all relation symbols are of arity at most two; see \cite[page 5]{GradelKV97} for a justification.

Two $\esotwo$-sentences $\psi$ and $\psi'$ are \emph{$\calK$-equivalent} for a class of structures $\calK$, if $\frakA \models \psi$ if and only if $\frakA \models \psi'$ for all structures $\frakA \in \calK$. An $\esotwo$-sentence is \emph{satisfiable on a class of structures $\calK$} if it has a model $\frakA\in \calK$.  We will also say that a sentence is \emph{$\calK$-satisfiable}.

Every $\esotwo$-formula can be translated into a formula of a very simple shape which is $\calK$-equivalent (see e.g. \cite{Scott1962} and  \cite[page 17]{GraedelO1999}). In this article $\calK$ will always be a class of ordered $O$-structures.

\begin{lemma}[Scott Normal Form]\label{lemma:scottnormalform}
  For every $\esotwo$-formula~$\varphi = \exists \bar R\, \xi$ there is a $\calK$-equivalent $\esotwo$-formula~$\varphi' = \exists \bar R'\, \xi'$ such that the formula $\xi'$ is of the form
  \[ \forall x \forall y \psi(x,y) \wedge \bigwedge_i \forall x \exists y \psi_i(x, y) \]
  where $\psi$ and all $\psi_i$ are quantifier-free formulas.  Moreover, $\varphi'$ can be computed in polynomial time.
\end{lemma}

Note that this lemma is slightly stronger than the usual Scott normal form lemma because interpretations for some relation symbols are restricted by $\calK$.  The usual proof carries over since interpretations of symbols that occur in $\varphi$ are preserved while going from a model of $\varphi$ to a model of $\varphi'$ and vice versa.

In this work it will be convenient to use an even stronger normal form. Our plan is to rephrase the satisfiability problem for $\esotwo$ into the satisfiability of a set of existential and universal constraints. This is a simple generalization of the approach taken in \cite{SchwentickZ12}. 

We need the following notions. Let $T$ be a signature. A \emph{unary $T$-type} is a maximally
consistent conjunction of unary $T$-literals using variable $x$
only. A \emph{(strictly) binary $T$-type} is a maximally consistent
conjunction of binary $T$-literals using variables $x$ and $y$, where
the conjunction includes the conjunct $x \neq y$.
Each element of a structure has a unique unary $T$-type and each pair of elements has a unique binary $T$-type.
If $\gamma$ is the binary type of $(a,b)$, then we denote the type of $(b,a)$ by $\bar \gamma$.
The set of unary and binary $T$-types are denoted by $\Sigma_T$ and $\Gamma_T$.

A \emph{constraint problem over $T$} is a tuple $C = (C_\exists,
C_\forall)$ where $C_\exists$ is a set of existential constraints and
$C_\forall$ is a set of universal constraints. An {\em
  existential constraint} $c_\exists$ is a tuple $(\sigma, E)$ where
$\sigma \in \Sigma_T$ and $E\subseteq \Gamma_T \times \Sigma_T$. A structure with
domain~$D$ satisfies $c_\exists$ if for every element $a \in D$ of
unary type~$\sigma$ there is a $(\gamma, \tau) \in E$ and an element
$b\in D$ of unary type $\tau$ such that $(a,b)$ has binary type
$\gamma$. A {\em universal constraint} $c_\forall$ is a tuple
$(\sigma, \gamma, \tau)$ where $\sigma, \tau \in \Sigma_T$ and $\gamma \in
\Gamma_T$. A structure with domain~$D$ satisfies $c_\forall$ if no
tuple $(a,b)\in D^2$ has binary type $\gamma$ if $a$ and $b$ have
unary types $\sigma$ and $\tau$.
A structure is a \emph{solution} of $C$ if it satisfies all constraints in $C_\exists$ and $C_\forall$. The problem $C$ is solvable if it has a finite solution. The \emph{size} $|C|$ of $C$ is~\mbox{$|\Sigma_T| + |\Gamma_T|+ |C_\forall|+ |C_\exists|$}. We emphasize that $C_\exists$ specifies required patterns whereas
$C_\forall$ specifies forbidden patterns.  

Let $\calK$ be class of structures over a signature $T' \subseteq T$. The problem $C$ is \emph{solvable on the class of structures $\calK$} if it has a solution $\frakA$ such that the restriction of $\frakA$ to $T'$ belongs to $\calK$. Such a solution will also be called \emph{$\calK$-solution}.

The following lemma shows that $\esotwo$-sentences can be translated into existential second-order constraint problems preserving satisfiability.

\begin{lemma}\label{lemma:formulaToConstraint}
  For every $\esotwo$-sentence $\varphi$ there is a constraint problem $C$ such that $\varphi$ is $\calK$-satisfiable if and only if $C$ has a $\calK$-solution. The constraint problem $C$ can be computed in exponential time in~$|\varphi|$.
\end{lemma}
  \begin{proof}
    Let $\varphi$ be an $\esotwo$-sentence. Without loss of generality we can assume, by Lemma \ref{lemma:scottnormalform}, that $\varphi = \exists \bar R \, \xi$ where $\xi$ is of the form $\forall x \forall y \psi(x,y) \wedge \bigwedge_i \forall x \exists y \psi_i(x,y)$ with quantifier-free $\psi$ and $\psi_i$. Let $T$ be the signature of the $\fotwo$-formula $\xi$. Further let $\Sigma$ and $\Gamma$ be the sets of unary and binary $T$-types.

  A $\calK$-satisfiability equivalent constraint problem for $\varphi$ is now constructed by translating $\forall x \forall y \psi(x,y)$ and $\forall x \exists y \psi_i(x,y)$ into universal and existential constraints, respectively. 
  
  Observe that, for a quantifier-free formula in disjunctive normal form $\bigvee_i \psi_i$, an equivalent quantifier-free formula in disjunctive normal form can be constructed, where (1) each disjunct is of the form $\sigma(x) \wedge \tau(y) \wedge \gamma(x,y)$ where $\sigma$ and $\tau$ are unary types and $\gamma$ is a conjunction of binary literals that use variables $x$ and $y$, and (2) for all unary types $\sigma$ and $\tau$ there is such a disjunct.
  
  Thus the first conjunct of $\xi$ is equivalent to a formula
    $$\chi = \forall x \forall y \bigvee_{\sigma, \tau \in \Sigma} \Big(\sigma(x) \wedge \tau(y) \wedge \gamma_{\sigma, \tau}(x,y)\Big)$$
        The formula $\chi$ is equivalent to the following formula:
    $$\mu = \forall x \forall y \bigwedge_{\sigma, \tau \in \Sigma}\Big( \sigma(x) \wedge \tau(y) \rightarrow \gamma_{\sigma, \tau}(x,y)\Big)$$
    To see this, let us consider a model $\frakA$ of $\chi$. For all elements $a, b$ of $\frakA$ of types $\sigma$, $\tau$, respectively, one of the disjuncts has to be satisfied. Since each element satisfies exactly one unary type, the disjunct $\sigma(x) \wedge \tau(y) \wedge \gamma_{\sigma, \tau}(x,y)$ is satisfied and therefore also $\mu$ is satisfied. 

    Now, consider a model $\frakA$ of $\mu$. Then for all elements $a, b$ of $\frakA$ the conjunction is satisfied. In particular the conjunct $\sigma(x) \wedge \tau(y) \rightarrow \gamma_{\sigma, \tau}(x,y)$ is satisfied, where $\sigma$ and $\tau$ are the types of $a$ and $b$, respectively. Therefore also $\sigma(x) \wedge \tau(y) \wedge \gamma_{\sigma, \tau}(x,y)$ and hence $\chi$ are satisfied.

    The formula $\mu$ can be easily translated into a set of universal constraints $C_\forall$. For all unary types $\sigma$ and $\tau$, the set $C_\forall$ contains a constraint $(\sigma, \gamma, \tau)$ for each binary type $\gamma$ that is not consistent with $\gamma_{\sigma, \tau}(x,y)$.
    Next we show how to translate the second part of $\xi$ into existential constraints. As before, we translate $\psi_i$ into the disjunctive normal form $\bigvee_{\sigma, \tau \in \Sigma}\big(\sigma(x) \wedge \tau(y) \wedge \gamma_{\sigma, \tau}(x,y)\big)$. Sorting by $\sigma$ and moving the existential quantifier inwards as far as possible yields that every conjunct of the second part of $\psi$ can be written as follows:
    $$\chi_i = \forall x \bigvee_{\sigma\in \Sigma} \Big(\sigma(x) \wedge \exists y \bigvee_{\tau\in \Sigma}\big( \tau(y) \wedge \gamma_{\sigma, \tau}(x,y)\big)\Big)$$
    Using a similar argument as above one can show that each $\chi_i$ is equivalent to a formula:
    $$\mu_i = \forall x \bigwedge_{\sigma\in \Sigma} \Big(\sigma(x) \rightarrow \exists y \bigvee_{\tau\in \Sigma}\big( \tau(y) \wedge \gamma_{\sigma, \tau}(x,y)\big)\Big)$$
    The formulas $\mu_i$ can be easily translated into a set of existential constraints $C_\exists$. For all $i$ and all $\sigma$, the set $C_\exists$ contains an existential constraint $(\sigma, E)$ where $E$ contains all tuples $(\gamma, \tau)$ that are consistent with $\gamma_{\sigma, \tau}$ where $\tau$ is a unary type and $\gamma$ is a binary type.

    Observe that the sets $C_\exists$ and $C_\forall$ can be computed in exponential time (where the exponential blow-up comes from using the types).
  \end{proof}

The following notion will be used frequently. If an element has a unary type $\sigma$ we also say that it is \emph{$\sigma$-labeled}.
Consider a constraint problem $C$ over $T$, a $T$-structure $\frakA$, a $\sigma$-labeled element~$a$ and an existential constraint $(\sigma, E)$. An element $b$ is a \emph{$(\sigma, E)$-witness} of $a$ if there is a $(\gamma,\tau) \in E$, the element $b$ has unary type $\tau$ and $(a,b)$ has $T$-type $\gamma$. A \emph{witness} of $a$ is an element that is a $(\sigma, E)$-witness for some existential constraint.

The following notations will be convenient for studying finite satisfiability of $\esotwo$ on ordered structures.  A constraint problem has a (finite) \emph{$(<_{1},<_{2})$-solution} if it is solvable on the class of finite ordered $(<_{1},<_{2})$-structures. We use similar definitions for $\esotwo(\suc_1, \suc_2)$ and so on.

For a constraint problem over $O \cup T$ where $O$ is a subset of the symbols $<,<_1, <_2, \ldots$ and $\suc, \suc_1, \suc_2, \ldots$ we will denote existential and universal constraints in a slightly different fashion. To this end observe that all unary $O$-types are trivial (e.g. $x < x \wedge \lnot \suc(x,x)$) and can, without loss of generality, be omitted. Each binary $T \cup O$-type is a conjunction of a binary $O$-type and binary $T$-type. Therefore we can write existential constraints of constraint problems over $T \cup O$ as tuples $(\sigma, E)$ where $\sigma \in \Sigma_{T}$ and $E$ is a set of tuples $(d, \gamma, \tau)$ with $d \in \Gamma_O$, $\gamma \in \Gamma_{T}$ and $\tau \in \Sigma_{T}$. Similarly universal constraints are written as tuples $(\sigma, d, \gamma, \tau)$ where  $\sigma, \tau \in \Sigma_{T}$, $\gamma \in \Gamma_{T}$ and $d \in \Gamma_O$.

\section{Warm-Up: One Linear Order and One Successor}\label{section:warm-up}

Before proving the main results we study the complexity of the $\esotwo$-satisfiability problem on ordered $(\suc,<)$-structures. Decidability has been established in \cite{GradelR99}. Combining Lemma 2.4 in \cite{GradelR99} and a construction from \cite[Corollary 9.2]{EiterGG00} yields at best a non-deterministic double-exponential upper bound. In this section we obtain an optimal non-deterministic exponential upper bound.
\begin{theorem}\label{theorem:successorandorder}
 $\esotwo$-satisfiability on finite ordered $(\suc,<)$-structures is \NEXPTIME-complete.
\end{theorem}

This result should be compared with the known results that deciding \mbox{$\emsotwo$} on ordered $(\suc,<)$-structure and deciding $\esotwo$ on ordered $<$-structures are $\NEXPTIME$-complete~\mbox{\cite{EtessamiVW02,Otto01}}.

We present two different approaches for proving Theorem \ref{theorem:successorandorder}; a small model property based approach and an automata-based approach. Only the former approach leads to a non-deterministic exponential time algorithm. Later both approaches will be generalized in different directions to obtain decidability for larger fragments.  The small model approach, combined with a technique from \cite{SchwentickZ10,SchwentickZ12}, is used to obtain a decision algorithm for $\esotwo$ on $(\suc_1, <_1, <_2)$-structures in Section \ref{section:twoorders}. The automata-based approach is used to solve the $\esotwo$-satisfiability problem on $(\suc_1, \suc_2, <_2)$-structures in Section \ref{section:twosuccessors}. 

To decide whether an $\esotwo$-sentence $\varphi$ has a model which is a finite $(\suc,<)$-structure, $\varphi$ is converted into a satisfiability equivalent constraint problem $C$ whose size is exponential in the size of $\varphi$ using Lemma~\ref{lemma:formulaToConstraint}. We show, in Lemma~\ref{lemma:succordfinitemodel}, that if $C$ has a finite $(\suc,<)$-solution then it has a $(\suc,<)$-solution of size at most $N$ where $N$ is polynomial in the size of $C$. Hence, a non-deterministic exponential time algorithm can guess a structure of size at most $N$ and verify that it is indeed a solution of $C$. This proves Theorem~\ref{theorem:successorandorder}. In the automata-based approach the satisfiability of $C$ is checked by a finite state automaton, see Lemma~\ref{lemma:succorder}.

The following notion will be useful in both approaches. Let $w$ be a word. A position $a$ of $w$ is called \emph{$(\sigma, \tau, k)$-rich} if there are at least $k$ $\sigma$-labeled positions (strictly) before $a$ and at least $k$ $\tau$-labeled positions (strictly) after $a$.  A position $a$ of $w$ is \emph{$(\sigma, \tau, k)$-poor} if there are at most $k$ $\sigma$-labeled positions (strictly) before $a$ and at most $k$ $\tau$ labeled positions (strictly) after $a$. We stress that if a position is not $(\sigma, \tau, k)$-rich than it is not necessarily $(\sigma, \tau, k)$-poor, and vice versa.

\begin{lemma}\label{lemma:richpoor}
  If a word has no $(\sigma, \tau, k)$-rich position then it has a $(\sigma, \tau, k+1)$-poor position.
\end{lemma}

\begin{proof}
  Let $w$ be a word with no $(\sigma, \tau, k)$-rich position. Towards a contradiction assume that $w$ has no $(\sigma, \tau, k+1)$-poor position. Then, by definition, for each position $i$ of $w$ there are at least $k+1$ $\sigma$-labeled positions strictly left of $i$ or at least $k+1$ $\tau$-labeled positions strictly right of $i$.

  In particular there are at least $(k+1)$ $\sigma$-labeled positions to the left of the last position of $w$. Let $i$ be the $(k+1)$th $\sigma$-labeled position of $w$. Since $i$ has only $k$ $\sigma$-labeled positions to its left, it must have at least $k+1$ $\tau$-labeled positions to is right. Hence, $i$ is $(\sigma, \tau, k)$-rich; a contradiction.
\end{proof}

As outlined above, the proof of Theorem~\ref{theorem:successorandorder} follows immediately from the following small model property and Lemma~\ref{lemma:formulaToConstraint}. The proof of the small model property is in the same spirit as the proof of Lemma 2.3 in \cite{GradelR99}.
\begin{lemma}\label{lemma:succordfinitemodel}
  If a constraint problem $C$ has a finite $(\suc, <)$-solution then it has such a solution of size polynomial in $|C|$.
\end{lemma}
\begin{proof}   Assume that $C \df (C_\exists, C_\forall)$ is a constraint problem over a signature $T \cup \{\suc, <\}$. Let $\Sigma \df \Sigma_T$ and $\Gamma \df \Gamma_T$. Let $k \df 3 |\Gamma|$.

  We show that if $C$ has a finite $(\suc,<)$-solution $\frakA$ with at least $N$ elements then it also has a $(\suc,<)$-solution $\frakB$ with $|\frakB| < N$. The number $N$ is polynomial in $|C|$ and will be specified later.

  For each element $a$ of $\frakA$ denote by $W(a)$ a set containing witnesses for $a$ for each existential constraint. An element $b \in W(a)$ is a \emph{local witness} if $\suc(a,b)$, $a=b$, or $\suc(b,a)$. Otherwise it is a \emph{non-local witness}.

  We first construct, from $\frakA$, a solution $\frakA'$ of $C$ whose domain and unary types coincide with the domain and unary types of $\frakA$ but whose set of non-local witnesses is small. Afterwards we argue that a smaller solution $\frakB$ can be obtained from $\frakA'$ by removing some elements.
  
  Towards constructing $\frakA'$ let $Z_1$ be the set that contains, for every $\sigma \in \Sigma$, the first $k+1$ $\sigma$-positions and the last $k+1$ $\sigma$-positions of $\frakA$ (if these exist).  Further let $Z_2$ be a set that contains a witness for each position in $Z_1$ and each existential constraint, and let $Z \df Z_1 \cup Z_2$. Observe that $Z_1$ and $Z_2$ are of size polynomial in $|\Sigma| |\Gamma|$.

  We reassign some of the binary $T$-types of $\frakA$ in order to obtain~$\frakA'$. The goal is that every element $a$ in $\frakA'$ has a set $W'(a)$ of witnesses for $C_\exists$ such that (1) $W(a)$ and $W'(a)$ coincide with respect to local witnesses, and (2) all non-local witnesses in $W'(a)$ are from $Z$. To this end let $\sigma, \tau \in \Sigma$ and $d \df x < y \wedge \neg \suc(x,y)$. Denote $\bar d \df   y < x \wedge \neg \suc(y,x)$.
  
  If there is a $(\sigma, \tau, k)$-rich position $u$ in $\frakA$ then we reassign the binary types of all $\sigma$- and $\tau$-positions $a$ and $b$ satisfying $d$ by using a technique employed by Otto for constructing small models for satisfiable $\fotwo(<)$-sentences \cite{Otto01}. For completeness we recall the construction. Let $A = A_1 \cup A_2 \cup A_3$ with disjoint $A_i$ and $|A_i| = |\Gamma|$ contain the first $k$ $\sigma$-labeled elements. Similarly let $B = B_1 \cup B_2 \cup B_3$ with disjoint $B_i$ and $|B_i| = |\Gamma|$ contain the last $k$ $\tau$-labeled elements. Assume that $\Gamma \df \{\gamma_1, \ldots, \gamma_m\}$. Then: 
    
\begin{enumerate}
      \item[(A1)] Witnesses for elements in $A \cup B$ are assigned as follows:
        \begin{enumerate}
          \item If $a \in A_i$ is a $\sigma$-labeled element that has a $(d, \gamma_\ell, \tau)$-witness $b \in W(a)$ in $\frakA$ then the binary $T$-type of $(a, b')$ is set to $\gamma_\ell$ where $b'$ is the $\ell$th element of $B_i$. The element $b'$ is the intended $(d, \gamma_\ell, \tau)$-witness of $a$ in $\frakA'$.
          \item If $b \in B_i$ is a $\tau$-labeled element that has a $(\bar d, \gamma_\ell, \sigma)$-witness $a \in W(b)$ in $\frakA$ then the binary $T$-type of $(b, a')$ is set to $\gamma_\ell$ where $a'$ is the $\ell$th element of $A_{i+1}$ (where $i+1$ is calculated modulo 3).  The element $a'$ is the intended $(\bar d, \gamma_\ell, \tau)$-witness of $b$ in~$\frakA'$.
        \end{enumerate}
      \item[(A2)] Witnesses for all other tuples of $\sigma$- and $\tau$-labeled elements are assigned as follows:
        \begin{enumerate}
          \item If $a \notin A$ is a $\sigma$-labeled element that has a $(d, \gamma_\ell, \tau)$-witness $b \in W(a)$ in $\frakA$ then the binary $T$-type of $(a, b')$ in $\frakA'$ is $\gamma_\ell$ where $b'$ is the $\ell$th element of $B_1$. The element $b'$ is the $(d, \gamma_\ell, \tau)$-witness of $a$ in $\frakA'$.
          \item If $b \notin B$ is a $\tau$-labeled element that has a $(\bar d, \gamma_\ell, \sigma)$-witness $a \in W(b)$ in $\frakA$ then the binary $T$-type of $(b, a')$ in $\frakA'$ is $\gamma_\ell$ where $a'$ is the $\ell$th element of $A_{1}$. The element $a'$ is the $(d, \gamma_\ell, \tau)$-witness of $b$ in $\frakA'$.
        \end{enumerate}
      \item[(A3)] If a tuple $(a, b)$ with $\sigma$-labeled $a$, $\tau$-labeled $b$ and satisfying $d$ has not been assigned a binary $T$-type in $\frakA'$ so far, then it inherits its type from $\frakA$.
    \end{enumerate}
    
  Observe that no conflicts arise from (A1) and (A2). Furthermore, for all $\sigma$- and $\tau$-labeled elements $a$ there is a set $W'(a)$ satisfying (1) and (2). Moreover, no conflicts with universal constraints arise since no new types have been created. This concludes the case when there is a $(\sigma, \tau, k)$-rich position in $\frakA$.
  
  If there is no $(\sigma, \tau, k)$-rich position $u$ in $\frakA$ then Conditions (1) and (2) are already satisfied for all $\sigma$- and $\tau$-labeled positions $a$ and $b$ satisfying $d = x < y \wedge \suc(x,y)$. To see this we argue as follows. By Lemma \ref{lemma:richpoor} there is a $(\sigma, \tau, k+1)$-poor position $v$. Now let $a$ be a $\sigma$-labeled  position. If $a \leq v$ then all $(d, \gamma, \tau)$-witnesses $b \in W(a)$ are contained in $Z$ by construction (as $a$ is one of the $k+1$ leftmost $\sigma$-labeled positions). If $a > v$ then all $(d, \gamma, \tau)$-witnesses of $a$ are among the last $k+1$ $\tau$-labeled positions which are also contained in~$Z$. The argument for $\tau$-labeled positions is symmetric.
  
  The structure $\frakA'$ thus constructed satisfies (1) and (2).
  
  It remains to construct $\frakB$. If $|\frakA'| > N$ for $N \geq c (|\Sigma||\Gamma|)^4$ for an appropriate constant $c$ then there are positions \mbox{$a_1, a_2 \notin Z$} with successors $b_1$ and $b_2$ such that (i) there is no position $z \in Z$ with $a_1 < z < a_2$, (ii) $a_1$ and $a_2$ have the same unary $T$-type, and (iii)  $(a_1, b_1)$ and $(a_2, b_2)$ have the same binary $T$-type. The solution~$\frakB$ is obtained from $\frakA'$ by removing all elements between $a_1$ and $a_2$ (including $a_2$), and assigning to $(a_1, b_2)$ the binary type of $(a_1, b_1)$ in $\frakA'$. Then $\frakB$ satisfies all universal constraints of $C$ since no new types have been created. Further $\frakB$ satisfies all existential constraints of $C$ since elements inherit their local and non-local witnesses from $\frakA'$, and $\frakA'$ only uses elements from $Z$ as non-local witnesses. 
  \end{proof}

  Now we present the automata-based approach. As discussed above, this approach only yields an exponential space algorithm, yet it will later be used as the basis for a decision algorithm for a larger fragment.

\begin{lemma}\label{lemma:succorder}
  For every constraint problem $C$ there is a non-deterministic finite state automaton $\calA$ such that $C$ has a finite $(\suc,<)$-solution if and only if $L(\calA)$ is non-empty.
\end{lemma}
Here $L(\calA)$ denotes the language recognized by $\calA$.
\begin{proof}
  Assume that $C \df (C_\exists, C_\forall)$ is a constraint problem over a signature $T \cup \{\suc, <\}$. Let $\Sigma \df \Sigma_T$ and $\Gamma \df \Gamma_T$. 

  Without loss of generality we assume that the witnesses requested by the existential constraints of $C$ do not contradict universal constraints. More precisely, for an existential constraint $(\sigma, E)$ and every $(d, \gamma, \tau) \in E$ we assume that there is no universal constraint~$(\sigma, d, \gamma, \tau)$. 

  We construct a finite state automaton $\calA$ such that $\calA$ accepts a word over $\Sigma$ if and only if $C$ has a solution. Intuitively the automaton $\calA$ interprets words as $(\suc, <)$-structures with no binary relations from $T$. In order to accept a word $w$, it has to verify that binary types can be assigned to all pairs of positions in a way consistent with $C$.

  The main difficulty in verifying the existence of an assignment of binary types is to ensure that the types of tuples $(a,b)$ and $(b, a)$ are consistent. More precisely, if $(a,b)$ has to be typed $\gamma \in \Gamma$ due to some existential constraint and if $(b,a)$ has to be typed $\gamma' \in \Gamma$ due to some other existential constraint, then $\gamma$ and $\gamma'$ have to be compatible, that is $\gamma' = \bar \gamma$.
  
  In the following we describe the construction of $\calA$, afterwards we argue that the construction is correct. For an easier exposition of the automaton $\calA$, we often assume that positions are labeled with some extra information. It can be easily verified that those labels could also be guessed by the automaton (or, alternatively, they could be contained in an extended alphabet).
  
    \begin{enumerate}
    \item[(E)] We assume that for every existential constraint $(\sigma, E)$, all $\sigma$-labeled positions $a$ are labeled with a fresh label $(\sigma, d, \gamma, \tau)$ such that $(d, \gamma, \tau) \in E$. The intention is that the $(\sigma, E)$-witness of $a$ satisfies $(d, \gamma, \tau)$.
  \end{enumerate}
  
  The automaton has to verify that binary $T$-types can be assigned such that the witnesses declared in (E) exist and all pairs of positions satisfy the universal constraints. To this end the automaton handles positions that are $<$-close to each other and positions that are $<$-remote from each other in a different way.
  
  Dealing with positions that are close to each other is simple. Each position $a$ has at most two positions that are $<$-close to it: there might be a position  $b_1$ with $\suc(a, b_1)$ and a position $b_2$ with~$\suc(b_2, a)$. The positions $b_1$ and $b_2$ might not exist (if $a$ is the first or last position with respect to $<$). For all $a$ the automaton can guess and verify the binary types for $(a, b_1)$ and $(a, b_2)$.
 
  \begin{enumerate}
    \item[(L1)] (Local types) We assume that each position $a$ is labeled by up to two labels $\gamma_1, \gamma_2 \in \Gamma$. The intention is that $\gamma_1$ is the binary $T$-type of $(a,b_1)$ (if the position $b_1$ exists), and likewise for $\gamma_2$.
    \item[(L2)] (Consistency of local types) The automaton verifies that the labels are consistent, that is, e.g., if $a$ and $b$ are positions with $\suc(a,b)$ then the label $\gamma_1$ of $a$ (i.e. the type guessed for~$(a, b)$) is compatible with the type $\gamma_2$ of $b$ (i.e. the type guessed for~$(b, a)$).
    \item[(L3)] (Local witnesses) For every $\sigma$-labeled position $a$ that is labeled by $(\sigma, d, \gamma, \tau)$ due to~(E), the automaton verifies that if \mbox{$d = \suc(x,y)$} then the label $\gamma_1$ of $a$ is $\gamma$ and that its successor is labeled with $\tau$. Likewise for $d = \suc(y,x)$.
    \item[(L4)] (Local universal constraints) For all $\sigma$- and $\tau$-labeled positions $a$ and $b$ with $\suc(a,b)$ the automaton verifies that if~$a$ is labeled~$\gamma_1$ then there is no universal constraint $(\sigma, \suc(x,y), \gamma_1, \tau)$. Likewise for $\suc(b,a)$. 
  \end{enumerate}

  Verifying the existence of a type assignment for positions that are remote from each other is more intricate. The automaton $\calA$ has to check that there is at least one binary $T$-type consistent with the universal constraints that can be assigned to every pair of far-away positions and that witnesses for the existential constraints can be assigned consistently according to (E). 
  
  The former condition can be verified easily (and analogously to~(L4)):
  \begin{enumerate}
   \item[(U)] For all $\sigma$- and $\tau$-labeled positions $a$ and $b$ satisfying $d = x < y \wedge \neg \suc(x,y)$ the automaton verifies that there is a $\gamma$ such that there is no universal constraint $(\sigma, d, \gamma, \tau)$. Similarly for $\bar d$.
  \end{enumerate}

  Testing that existential witnesses can be assigned to remote positions is more involved. We discuss how $\calA$ verifies that binary types for all $\sigma$- and $\tau$-labeled positions $a$ and $b$ satisfying \mbox{$d \df x < y \wedge \neg \suc(x,y)$} can be assigned; the other case is symmetric. Let $\bar d \df d(y,x)$. To verify that binary types can be assigned to such $\sigma$- and $\tau$-labeled positions, the automaton $\calA$ guesses whether there is a $(\sigma, \tau, k)$-rich position $u$ with~\mbox{$k \df 3(|\Gamma|+3)$}. 
  
  If such a position $u$ exists then the automaton tests that every $\sigma$-labeled position has as many $\tau$-labeled positions to its right as is required by (E), and that every $\tau$-labeled position has sufficiently many $\sigma$-labeled positions to its left. If this is the case then binary types can be assigned using an assignment technique similar to Otto's technique \cite{Otto01} for reducing the size of models for $\fotwo(<)$ sentences (see the correctness proof below). 
  
  If there is no $(\sigma, \tau, k)$-rich position $u$ then, by Lemma \ref{lemma:richpoor} there is a $(\sigma, \tau, k+1)$-poor position~$v$. The automaton exploits this structure to guess and verify the binary types in this case.
  
  More precisely, for all $\sigma$ and $\tau$ the automaton does the following:

  \begin{enumerate}
    \item[(E1)] It guesses whether there is a $(\sigma, \tau, k)$-rich position $u$. The correctness of the guess can be easily verified by $\calA$.
    \item[(E2)] If there is such a position $u$ then $\calA$ verifies the following:
      \begin{enumerate}
       \item If a $\sigma$-labeled position $a$ is labeled by $(\sigma, d, \gamma_1, \tau), \ldots,$ $(\sigma, d, \gamma_m, \tau)$ in (E) then there are $m$ $\tau$-labeled positions to the right of $a$.
       \item Symmetrically, if a $\tau$-labeled position $a$ is labeled by $(\tau, \bar d, \gamma_1, \sigma)$, $\ldots$, $(\tau, d, \gamma_m, \sigma)$ in (E) then there are $m$ $\sigma$-labeled positions to the right of $a$.
      \end{enumerate}
    \item[(E3)] If there is no $(\sigma, \tau, k)$-rich position then $\calA$ guesses a $(\sigma, \tau, k+1)$-poor position $v$.     Then: 
      \begin{enumerate}
	\item We assume that positions are labeled by the following extra information (using fresh labels depending on $\sigma$ and $\tau$):
	\begin{itemize}
	  \item The $i$th $\sigma$-labeled position $a$ to the left of $v$ is labeled by $\sigma_i$. The $i$th $\tau$-labeled position $a$ to the right of $v$ is labeled by $\tau_i$. (As $v$ is $(\sigma, \tau, k+1)$-poor, there are at most $k+1$ such $\sigma_i$ and $\tau_i$.)
	  \item The intended witnesses for $\sigma_i$- and $\tau_i$-labeled positions are labeled as follows:
	    \begin{itemize}
	     \item If the $\sigma_i$-labeled position $a$ is labeled by $(\sigma, \gamma, d, \tau)$ in (E), then there is a $\tau$-labeled position $b$ which is also labeled with $(\sigma_i, \gamma, d, \tau)$ such that $(a,b)$ satisfies $d$. 
	     \item Likewise, if the $\tau_i$-labeled position $a$ is labeled by $(\tau, \gamma, \bar d, \sigma)$ in (E), then there is a $\sigma$-labeled position $b$ which is also labeled with $(\tau_i, \gamma, \bar d, \sigma)$ such that $(a,b)$ satisfies~$\bar d$. 
	    \end{itemize}
	  \item Positions with intended witnesses from $\sigma_i$- and $\tau_i$-labeled positions are labeled as follows:
	    \begin{itemize}
	      \item  Each position $a$ to the left of $v$ that is labeled by $(\tau, \gamma, \bar d, \sigma)$ in (E) is also labeled with $(\tau, \gamma, \bar d, \sigma_i)$ for some $i$ such that the tuple $(a, b)$ satisfies~$\bar d$ where $b$ is the $\sigma_i$-labeled position $b$ .
	      \item  Likewise, each $(\sigma, \gamma, d, \tau)$-labeled position $a$ to the right of $v$ is labeled with $(\sigma, \gamma, d, \tau_i)$ for some $i$ such that  the tuple $(a, b)$ satisfies~$d$ where $b$ is the $\sigma_i$-labeled position $b$.
	    \end{itemize}
	\end{itemize}
	\item The automaton verifies that the labels are consistent, that is:
	  \begin{itemize}
	    \item No $\tau$-labeled position to the left of $v$ is labeled with $(\sigma_i, \gamma, d, \tau)$ and with $(\tau, \gamma', \bar d, \sigma_i)$ where $\gamma$ and $\gamma'$ are not reverse types.
	    \item Likewise, no $\sigma$-labeled position to the right of $v$ is labeled with $(\tau_i, \bar \gamma', d, \sigma)$ and with $(\sigma, \gamma, d, \tau_i)$ where $\gamma$ and $\gamma'$ are not reverse types.
	  \end{itemize}
      \end{enumerate}
  \end{enumerate}
  We emphasize again that all labels assigned in (E3) can also be guessed by $\calA$.

  We argue that the construction of $\calA$ is correct. Clearly, if the set $C$ of constraints is satisfiable by a $(\suc, <)$-structure then the corresponding word is accepted by $\calA$ (the labels in (E), (L1), (E2) and (E3a) are assigned according to the solution of $C$).
  
  Now, assume that $\calA$ accepts a word $w$ and let $\rho$ be an accepting run. We argue that a $(\suc, <)$-structure $\frakA$ that satisfies $C$ can be obtained from $w$. The elements of $\frakA$ are the positions of $w$ ordered as in $w$ and the unary type of an element $v$ is its label in~$w$.
  
  We describe how to assign binary types. To this end we first assign, for every position $a$, types witnessing the existential constraints. Afterwards all so far non-typed pairs are typed by some type admissible by the universal constraints. 
  
  For each pair $(a,b)$ of nodes with $\suc(a,b)$, the edge $(a,b)$ in $T$ is typed with the type guessed in (L1). No binary type conflicts arise from this due to (L2).
  
  For non-local pairs of nodes, binary types of witnesses are assigned simultaneously for all $\sigma$- and $\tau$-labeled nodes. In the following assume that $a$ and $b$ are $\sigma$- and $\tau$-labeled nodes and that $(a,b)$ satisfies $d \df x < y \wedge \neg \suc(x,y)$,

  If there is a position $u$ as in (E1) then the witnesses can be assigned as in \cite{Otto01}. For the sake of completeness we recall the assignment strategy. Let $A$ be the first $k$ $\sigma$-labeled nodes and assume that $A = A_1 \cup A_2 \cup A_3$ with disjoint $A_i$ and $|A_i| = |\Gamma|$. Likewise let $B = B_1 \cup B_2 \cup B_3$ be the last $k$ $\tau$-labeled nodes with disjoint $B_i$ and $|B_i| = |\Gamma|$. Further assume that $\Gamma \df \{\gamma_1, \ldots, \gamma_m\}$.

  Then the witness types are assigned as follows:
  \begin{enumerate}
    \item[(A1)] The witnesses for elements $A$ and $B$ are assigned as follows:
      \begin{enumerate}
	\item If a position $a \in A_i$ is labeled by $(\sigma, d, \gamma_\ell, \tau)$ in $\rho$ then its $\gamma_\ell$-witness is the $\ell$th element $b$ of $B_i \setminus N(a)$ where $N(a)$ contains $a$ and all elements that are close to $a$.
	\item If a position $b \in B_i$ is labeled by $(\tau, d, \gamma_{\ell}, \sigma)$ in $\rho$ then its $\gamma_{\ell}$-witness is the $\ell$th element of $A_{i+1} \setminus N(b)$ where $i+1$ is calculated modulo~$3$.
      \end{enumerate}
    \item[(A2)]  The witnesses for elements not in $A$ and $B$ are assigned as follows:
      \begin{enumerate}
	\item  If a $\sigma$-labeled element $a \notin A$ is also labeled with $(\sigma, d, \gamma_\ell, \tau)$ then its $\gamma_\ell$-witness is the $\ell$-last element of $B \setminus N(a)$ (which is to the right of $a$ due to (E2)).
	\item  If a $\tau$-labeled element $a \notin B$ is also labeled with $(\tau, d, \gamma_\ell, \sigma)$ then its $\gamma_\ell$-witness is the $\ell$th element in $A \setminus N(a)$.
      \end{enumerate}
      \end{enumerate}
  
  The type assignments from (A1) ensure that all elements in $A$ and $B$ have their witnesses. The type assignments from (A2) ensure that all other elements have their witnesses. Observe that for $a \neq b$ at most one of the pairs $(a,b)$ and $(b,a)$ is assigned a type in (A1) and (A2); this determines the type of the other pair.

  If there is no position $u$ as in (E1) then the automaton guessed a position $v$ in run $\rho$. The witness types are assigned according to the labels assumed in (E3a), that is: 
  \begin{enumerate}
   \item[(A'1)] The witnesses for the at most $k+1$ $\sigma$-labeled elements to the left of $v$ and the at most $k+1$ $\tau$-labeled elements to the right of $v$ are assigned as follows:  
   \begin{itemize}
      \item If $a$ is a $\sigma$-labeled position to the left of $v$ that is labeled with $(\sigma, d, \gamma, \tau)$ and $\sigma_i$ in $\rho$ and if $b$ is the position labeled with $(\sigma_i, d, \gamma, \tau)$ in $\rho$ then $(a,b)$ is labeled $\gamma$.
      \item Likewise, if $a$ is a $\tau$-labeled position to the right of $v$ that is also labeled with $(\tau, \bar d, \gamma, \sigma)$ and $\tau_i$ in $\rho$ and if $b$ is the position labeled with $(\tau_i, d, \bar \gamma, \tau)$ in $\rho$ then $(a,b)$ is labeled~$\bar \gamma$.
   \end{itemize}
   \item[(A'2)] The witnesses for the other  elements are assigned as follows:  
   \begin{itemize}
      \item If $a$ is a $\tau$-labeled position to the left of $v$ that is labeled with $(\tau, \bar d, \gamma, \sigma_i)$ in $\rho$ then $(a,b)$ is labeled with $\gamma$ where $b$ is the position labeled with $\sigma_i$ in $\rho$.
      \item Likewise, if $a$ is a $\sigma$-labeled position to the right of $v$ that is labeled with $(\sigma, \bar d, \gamma, \tau_i)$ in $\rho$ then $(x,y)$ is labeled with $\bar \gamma$ where $b$ is the position labeled with $\tau_i$ in $\rho$.
   \end{itemize}
     \end{enumerate}
  Note that the assignments are consistent due to (E3c).
  
  Thus, so far types have been assigned such that all elements have witnesses, ensuring that the existential constraints are satisfied. All remaining so far non-typed edges are typed by some type admissible by the universal constraints. Such types exist due to condition (U).
\end{proof}
 
\section{Two Linear Orders and One Successor}\label{section:twoorders}
In this section we will show that $\esotwo$-satisfiability problem on finite $(\suc_{1},<_{1},<_{2})$-structures is decidable.
 It is known that $\esotwo$ is decidable on $<$-structures \cite{Otto01} and that $\emsotwo$ is decidable on finite $(<_{1},<_{2})$-structures.  We combine the approaches used for those two results as well as the technique introduced in Lemma \ref{lemma:succordfinitemodel} to obtain a nondeterministic double-exponential upper bound. We conjecture that this bound can be improved to exponential space by generalizing the methods from \cite{SchwentickZ12}.  

\begin{theorem}\label{thm:twoorders-decidable}
  $\esotwo$-Satisfiability on finite ordered $(\suc_1, <_1, <_2)$-structures is in $\NTWOEXPTIME$.
\end{theorem}

The result immediately follows from the following small solution property and Lemma \ref{lemma:formulaToConstraint}.

\begin{lemma}\label{lemma:twoordfinitemodel}
  If a constraint problem $C$ has a finite $(\suc_1, <_1, <_2)$-solution then it has such a solution of size exponential in $|C|$.
\end{lemma}
\begin{proof}
  We follow the proof outline employed for $\emsotwo$ in \cite{SchwentickZ10}; yet for dealing with binary relations the individual steps have to be generalized. For consistency with the proofs in  \cite{SchwentickZ10} we prove the statement of the lemma for  $(<_1, \suc_2,<_2)$-constraint problems. For the description below, recall that structures with two linear orders can be viewed as point sets in the plane (cf. Section~\ref{sec:ordered-structures}).

  In order to establish the exponential solution property for \mbox{$(<_1, \suc_2,<_2)$}-constraint problems, we show that smaller solutions can be constructed from large solutions. To this end we assign a \emph{profile} $Pro(c)$ to each element $c$ of a solution $\frakA$ in such a way that if the profiles $Pro(c_1)$ and $Pro(c_2)$ of two elements $c_1$ and $c_2$ of $\frakA$ with $c_1 <^\frakA_2 c_2$ coincide then a solution $\frakB$ with fewer elements can be constructed from $\frakA$ by deleting all elements $a$ with $c_1 <^\frakA_2 a \leq^\frakA_2 c_2$ and shifting the elements $b$ with $c_2 <^\frakA_2 b$ along the $<_1^{\frakA}$-dimension. Deleting elements might, of course, also eliminate some witnesses of remaining elements. The profiles of elements will be defined such that new witnesses can be assigned.

  We make this more precise now. Assume that $C$ is a \mbox{$(<_1, \suc_2,<_2)$}-constraint problem with signature $T \cup \{<_1, \suc_2, <_2\}$ and let $\Sigma \df \Sigma_{T}$ and $\Gamma \df \Gamma_{T}$. Let $k \df 3 |\Gamma|$. Further  let $\frakA$ be a solution of~$C$ with domain $D$. For every element $a$ of $\frakA$, we also fix a set $W(a)$ of elements witnessing that the existential constraints of $C$ are satisfied for $a$.
  
  We extend the notion of profiles introduced in \cite{SchwentickZ10} to structures with arbitrary binary relations. The profile of an element $c$ shall capture all relevant information about elements below $c$ that have witnesses above $c$, and vice versa. Storing information for all such witnesses in the profiles is not possible, yet it turns out that less information is sufficient for being able to construct smaller solutions.
  
  For defining the profile of $c$, we first specify a set $P(c)$ of elements that are important for~$c$. Then, roughly speaking, the profile will be defined as the substructure of $\frakA$ induced by~$P(c)$. The set $P(c)$ contains $c$, its $<_2^{\frakA}$-successor, a set of special elements $A(c)$ as well as some of the witnesses $A_W(c)$ of elements in $A(c)$. The last two sets will be defined next.
    The set $A(c)$ is the union of the sets  $A_{\sigma, \min, \downarrow}(c)$, $A_{\sigma, \max, \downarrow}(c)$, $A_{\sigma, \min, \uparrow}(c)$, and $A_{\sigma, \max, \uparrow}(c)$, for all $\sigma \in \Sigma$, defined as follows:
  \begin{enumerate}
   \item The set $A_{\sigma,\min, \downarrow}(c)$ contains the $k+1$ leftmost $\sigma$-labeled points that are below $c$ (if they exist). 
   \item The set  $A_{\sigma,\min, \uparrow}(c)$ contains the $k+1$ leftmost $\sigma$-labeled points that are above~$c$ (if they exist), but excluding the $<_2^{\frakA}$-successor of $c$.  
   \item The set $A_{\sigma,\max, \downarrow}(c)$ contains the $k+1$ rightmost $\sigma$-labeled points that are below $c$ (if they exist). 
   \item The set  $A_{\sigma,\max, \uparrow}(c)$ contains the $k+1$ rightmost $\sigma$-labeled points that are above~$c$ (if they exist), but excluding the $<_2^{\frakA}$-successor of $c$.  
  \end{enumerate}

  Intuitively $A_W(c)$ contains the relevant witnesses of elements in $A(c)$.  It is the union of the sets $W^\tau_{\sigma,\min, \downarrow}(c)$, $W^\tau_{\sigma,\min, \uparrow}(c)$,  $W^\tau_{\sigma,\max, \downarrow}(c)$, and $W^\tau_{\sigma,\max, \uparrow}(c)$, for all unary types $\sigma, \tau \in \Sigma$, defined as follows:
  \begin{enumerate}
    \item The set $W^\tau_{\sigma,\min, \downarrow}(c)$ contains, for all $a \in A_{\sigma, \min, \downarrow}(c)$, all witnesses $b \in W(a)$ that are to the right of $a$ and above $c$.
    \item The set $W^\tau_{\sigma,\min, \uparrow}(c)$ contains, for all $a \in A_{\sigma, \min, \uparrow}(c)$, all witnesses $b \in W(a)$ that are to the right of $a$ and below $c$.
    \item The set $W^\tau_{\sigma,\max, \downarrow}(c)$ contains, for all $a \in A_{\sigma, \min, \downarrow}(c)$, all witnesses $b \in W(a)$ that are to the left of $a$ and above $c$.
    \item The set $W^\tau_{\sigma,\max, \uparrow}(c)$ contains, for all $a \in A_{\sigma, \min, \uparrow}(c)$, all witnesses $b \in W(a)$ that are to the left of $a$ and below $c$.
  \end{enumerate}

  The set of important points for $c$ is $P(c) \df \{c, s(c)\} \cup A(c) \cup A_W(c)$ where $s(c)$ is the unique element satisfying $\suc_2^{\frakA}(c, s(c))$. The profile $Pro(c)$ of $c$ is the structure $(\calP(c), c)$ where $\calP(c)$ is the substructure of $\frakA$ induced by $P(c)$. Observe that relation $\suc_2^{\calP(c)}$ is not necessarily a successor relation.   
  Now we show that if two profiles $Pro(c_1)$ and $Pro(c_2)$ of elements $c_1\neq c_2$ are isomorphic then a solution $\frakB$ with fewer elements can be constructed from $\frakA$. The domain of $\frakB$ is the set $D' \subseteq D$ which contains all elements $a$ with $a \leq^\frakA_2 c_1$ or $c_2 <^\frakA_2 a$  (assuming without loss of generality that~$c_1 <^\frakA_2 c_2$). We now describe how the relations $<^\frakB_1$, $<^\frakB_2$ and $\suc^\frakB_2$ as well as the interpretations of symbols from $T$ are constructed for $\frakB$.

  Towards defining $<^\frakB_1$, $<^\frakB_2$ and $\suc^\frakB_2$ in $\frakB$, we fix an embedding~$\theta$ that maps every element $u$ of~$\frakA$ to a point~$\theta(u) \in  \rat \times \rat$ such that $\theta(u_1).\mathsf{x} < \theta(u_2).\mathsf{x}$ if and only if $u_1 <^\frakA_1 u_2$, and $\theta(u_1).\mathsf{y} < \theta(u_2).\mathsf{y}$ if and only if $u_1 <^\frakA_2 u_2$. Here $<$ is the usual linear order on the rational numbers, and $p.\mathsf{x}$ and $p.\mathsf{y}$ denote the $\mathsf{x}$- and $\mathsf{y}$-component of a point $p \in \rat \times \rat$. 

  From $\theta$ we define an embedding $\theta'$ of the elements of $\frakB$ into~$\rat \times \rat$ which will be used to obtain the order relations on~$\frakB$. Intuitively $\theta'$ keeps the positions of elements below $c_1$ but shifts elements above $c_2$ along the $\mathsf{x}$-direction in order to make those points consistent with the profile of $c_1$. The embedding is defined as follows:
  \begin{enumerate}
    \item[(P1)] For all elements $u \in D'$ with $u \leq^\frakA_2 c_1$, define $\theta'(u) \df \theta(u)$.
    \item[(P2)] For all elements $u\in D'$ with $c_2 <^\frakA_2 u$ the embedding $\theta'(u)$ is defined as follows. Assume that $c_1^1, \ldots, c_1^n$ are the elements of $P(c_1)$ ordered by $<^\frakA_1$ and  $c_2^1, \ldots, c_2^n$ are the elements of $P(c_2)$ ordered by $<^\frakA_1$. Observe that $c_2^1 \leq u \leq c_2^n$ by the definition of~$P(c_2)$. Assume that $\theta(c_2^i).\mathsf{x} \leq \theta(u).\mathsf{x} \leq \theta(c_2^{i+1}).\mathsf{x}$. Then the position of $u$ in the new orders is obtained by shifting the element $u$ such that its $\mathsf{x}$-coordinate is between the $\mathsf{x}$-coordinates of the elements $c_1^i$ and $c_1^{i+1}$. More precisely $\theta'(u).\mathsf{y} \df \theta(u).\mathsf{y}$ and $\theta'(u).\mathsf{x}$ is defined as follows. If $u = c_2^i$ for some $i$ then $\theta'(u).\mathsf{x} = \theta(c_1^i).\mathsf{x}$. Otherwise,  
      $$\theta(c_1^i).\mathsf{x} + \frac{\theta(u).\mathsf{x}-\theta(c_2^i).\mathsf{x}}{\theta(c_2^{i+1}).\mathsf{x}-\theta(c_2^{i}).\mathsf{x})} (\theta(c_1^{i+1}).\mathsf{x} - \theta(c_1^{i}).\mathsf{x})$$  
    
  \end{enumerate}
  
  It might happen, that the embedding $\theta'$ maps an element $u$ to a point in $\rat \times \rat$ whose $\mathsf{x}$-coordinate is already used by another element $v$. This can be dealt with in the same way as in~\mbox{\cite[Lemma 8]{SchwentickZ10}}. The $\mathsf{x}$-coordinates are assigned sequentially: first the ones for (P1), then the remaining elements of $c_2^1, \ldots, c_2^n$, and finally all remaining elements for (P2). No conflicts can arise for the first two cases. If the designated $\mathsf{x}$-coordinate of an element $u$ assigned in the third case is already used by an element $v$ then the $\mathsf{x}$-coordinate of $u$ is shifted by a very small distance. Namely, if  $w$ is the element whose $\mathsf{x}$-coordinate is the largest with $\theta'(w).\mathsf{x} < \theta'(v).\mathsf{x}$ assigned so far then the new $\mathsf{x}$-coordinate of $u$ is shifted in between $\theta'(w).\mathsf{x}$ and $\theta'(v).\mathsf{x}$. Note that such a $w$ always exists due to the first two cases.  

  In $\frakB$, the symbol $<_1$ is interpreted by the linear order $<^{\frakB}_1$ such that $a <^{\frakB}_{1} b$ if and only if $\theta'(a).x < \theta'(b).\mathsf{x}$. The interpretation of $\suc^{\frakB}_{1}$ is induced by $<^{\frakB}_1$. The interpretations of $<_2$ and $\suc_2$ are defined accordingly, but with respect to the $\mathsf{y}$-coordinates. Observe that $<^\frakB_2$ is induced by $<^\frakA_2$ and that for all tuples $(a,b)$ with either $a \leq^\frakA_2 c_1$ and $b \leq^\frakA_2 c_1$, or  $c_1 <^\frakA_2 a$ and $c_1 <^\frakA_2 b$, the relative order of the elements in $\frakB$ is the same.
    
  The unary type of an element $a$ in $\frakB$ is the same as in $\frakA$. It remains to assign binary $T$-types for tuples of elements of $\frakB$.
  The binary type of tuples $(a,b)$ with either $a \leq^\frakB_2 c_1$ and $b \leq^\frakB_2 c_1$ or $c_1 <^\frakB_2 a$ and $c_1 <^\frakB_2 b$ is inherited  from $\frakA$. The tuple $(c_1, b)$ with $(c_1, b)\in\suc_{2}^{\frakB}$ inherits the type of the tuple $(c_1,u)$ with $(c_{1},u) \in \suc_{2}^{\frakA}$ in $\frakA$.
  
  We now explain how the binary types of all tuples $(a, b)$ with $a \leq^\frakB_2 c_1 <^\frakB_2 b$ and $(a,b)\notin \suc^\frakB_2$ are assigned. Our focus is on the case when $a <^\frakB_1 b$; the case $b <^\frakB_1 a$ is symmetric. In the following let $d \df x <_1 y \wedge x <_2 y \wedge \neg \suc_2(x,y)$.
        We simultaneously assign all binary $T$-types to all $\sigma$-labeled $a$ and $\tau$-labeled $b$ such that $(a,b)$ satisfies~$d$.

  Depending on the structure of the points in $A(c_1)$ we distinguish two cases. To this end let $a_1, \ldots, a_n$ be the elements of $A(c_1)$ ordered by $<^\frakA_1$. Denote by $w(c_1)$ the sequence $(\sigma_1, d_1), \ldots, (\sigma_n, d_n)$ where $\sigma_i$ is the unary type of $a_i$ and $d_i = \uparrow$ if $c_1 <^\frakA_2 a_i$, $d_i = \suc$ if $\suc(c_1, a_i)$, $d_i = \cdot$ if $c_1 = a_i$ and $d_i = \downarrow$ if~$a_i <^\frakA_2 c_1$. 
  
  If there is a $((\sigma, \downarrow), (\tau, \uparrow), k)$-rich position $u$ in $w(c_1)$ then we assign the binary $T$-types of all tuples $(a,b)$ satisfying $d$ in $\frakB$ using the technique employed by Otto and also used in Lemma \ref{lemma:succordfinitemodel}. Let \mbox{$A = A_1 \cup A_2 \cup A_3$} with disjoint $A_i$ and $|A_i| = |\Gamma|$ be the set that contains the first $k$ elements of $A_{\sigma, \min, \downarrow}(c_1)$ in $\frakB$. Similarly let $B = B_1 \cup B_2 \cup B_3$ with disjoint $B_i$ and $|B_i| = |\Gamma|$ contain the last $k$ elements of $A_{\tau, \max, \uparrow}(c_2)$ in $\frakB$. Then the binary types are assigned as in Lemma~\ref{lemma:succordfinitemodel}: 

  \begin{enumerate}
      \item[(A1)] Witnesses for elements in $A \cup B$ are assigned as follows:
        \begin{enumerate}
          \item If, in $\frakA$, an element $a \in A_i$ is $\sigma$-labeled and has a $(d, \gamma_\ell, \tau)$-witness $b \in W(a)$ then the binary $T$-type of $(a, b)$ in $\frakB$ is $\gamma_\ell$ where $b$ is the $\ell$th element of $B_i$. The element $b$ is the $(d, \gamma_\ell, \tau)$-witness of $a$ in $\frakB$.
          \item If, in $\frakA$, an element $b \in B_i$ is $\tau$-labeled and has $(\bar d, \gamma_\ell, \sigma)$-witness $a \in W(b)$ then the binary $T$-type of $(b, a)$ in $\frakB$ is $\gamma_\ell$ where $a$ is the $\ell$th element of $A_{i+1}$ (where $i+1$ is calculated modulo 3). The element $a$ is the $(d, \gamma_\ell, \tau)$-witness of $b$ in $\frakB$.
        \end{enumerate}
      \item[(A2)] Witnesses for all other tuples of $\sigma$- and $\tau$-labeled elements are assigned as follows:
        \begin{enumerate}
          \item If, in $\frakA$, an element $b \notin B$ is $\sigma$-labeled and has a $(d, \gamma_\ell, \tau)$-witness $b \in W(a)$ then the binary $T$-type of $(a, b)$ in $\frakB$ is $\gamma_\ell$ where $b$ is the $\ell$th element of $B_1$. The element $b$ is the $(d, \gamma_\ell, \tau)$-witness of $a$ in $\frakB$.
          \item If, in $\frakA$, an element $b \notin B$ is a $\tau$-labeled and has a $(\bar d, \gamma_\ell, \sigma)$-witness $a \in W(b)$ then the binary type of $(b, a)$ in $\frakB$ is $\gamma_\ell$ where $a$ is the $\ell$th element of $A_{1}$. The element $a$ is the $(d, \gamma_\ell, \tau)$-witness of $b$ in $\frakB$.
        \end{enumerate}
      \item[(A3)] If a tuple $(a,b)\in D' \times D'$ such that, in $\frakA$, 
            the element $a$ is $\sigma$-labeled, $b$ is $\tau$-labeled and $(a,b)$ satisfies $d$, and if $(a,b)$ has not been assigned a binary $T$-type so far, then a type is assigned as follows. Since $Pro(c_1) \cong Pro(c_2)$, there is a $\tau$-labeled element $b'\in D$ such that $(a,b')$ satisfies $d$ in $\frakA$. The tuple $(a, b)$ inherits its binary type from $(a,b')$ in $\frakB$.
    \end{enumerate}
    
    This concludes the binary type assignments in the case when there is a $((\sigma, \downarrow), (\tau, \uparrow), k)$-rich position $u$ in $w(c)$. Note that this case is settled without using the fact that the profiles contain the elements added by~$A_W$.

    If there is no $((\sigma, \downarrow), (\tau, \uparrow), k)$-rich position $u$ in $w(c_1)$ then we assign the binary $T$-types to $\sigma$-labeled  $a$ and $\tau$-labeled $b$ satisfying $d$ as explained below. 
    
    We observe that in $\frakA$ the $(d, \cdot, \tau)$-witnesses for all $\sigma$-labeled elements $a <^\frakA_2 c_1$ are in  $W^\tau_{\sigma, \min, \downarrow}(c_1) \cup A_{\tau, \max, \uparrow}(c_1)$. To see this we argue as in Lemma \ref{lemma:succordfinitemodel}. By Lemma \ref{lemma:richpoor} there is a \mbox{$((\sigma, \downarrow), (\tau, \uparrow), k+1)$}-poor position in $w(c_1)$. Let $v$ be the minimal such position.  Now let $a <^\frakA_2 c_1$ be a $\sigma$-labeled  position. If $a \leq^\frakA_1 v$ then all $(d, \cdot, \tau)$-witnesses $b \in W(a)$ are contained in $W^\tau_{\sigma, \min, \downarrow}(c_1)$ by construction (as $a$ is one of the $k+1$ $<^\frakA_1$-smallest $\sigma$-labeled elements below $c_1$). If $a >^\frakA_1 v$ then all $(d, \cdot, \tau)$-witnesses of $a$ are among the elements $A_{\tau, \max, \uparrow}(c_1)$ (as there are only at most $k+1$ $\tau$-labeled elements above $c_1$ and to the right of $a$). Similarly all $(\bar d, \cdot, \sigma)$-witnesses for all $\tau$-labeled elements $b >^\frakA_2 c_2$ are in~\mbox{$A_{\sigma, \min, \downarrow}(c_2) \cup W_{\tau, \max, \uparrow}(c_2)$}. 
    
     Let $\pi$ be an isomorphism of $(\calP(c_1), c_1)$ and $(\calP(c_2), c_2)$. We use the above observation to assign binary $T$-types as follows:

      \begin{enumerate}
        \item[(B1)] Witnesses for $\sigma$-labeled $a$ with $a <^\frakB_2 c_1$ are assigned as follows. If $b \in W(a)$ is a $(d, \gamma, \tau)$-witness of $a$ in $\frakA$ then $\pi(b)$ is the $(d, \gamma, \tau)$-witness of $a$ in $\frakB$. That is, the binary type of $(a, \pi(b))$ in $\frakB$ is $\gamma$.        \item[(B2)] Witnesses for $\tau$-labeled $b$ with $b >^\frakB_2 c_2 $ are assigned as follows. If $a \in W(b)$ is a $(\bar d, \gamma, \sigma)$-witness of $b$ in $\frakA$ then $\pi(a)$ is the $(d, \gamma, \tau)$-witness of $b$ in $\frakB$. That is, the binary type of $(b, \pi(a))$ in $\frakB$ is $\gamma$.         \item[(B3)] If $(a,b) \in D' \times D'$ has not been assigned a binary $T$-type so far, then a type is assigned as follows. Since $Pro(c_1) \cong Pro(c_2)$, there is a $\tau$-labeled element $b'\in D$ such that $(a,b')$ satisfies $d$ in $\frakA$. The tuple $(a, b)$ inherits its binary type in $\frakB$ from $(a,b)$.
      \end{enumerate}
      
    This concludes the binary type assignments in the case when there is no $((\sigma, \downarrow), (\tau, \uparrow), k)$-rich position $u$ in $w(c)$.
    
    We shortly argue why the construction is correct. The assignments in (A1) and (A2) as well as (B1) and (B2) ensure that no conflicting types are assigned to $(a, b)$ and $(b, a)$. No conflicts with universal constraints arise by the assignments from (A1)-(A3) and (B1)-(B3), as no new types are introduced due to the choice of $\theta'$. Finally, after the assignments (A1)-(A2) and (B1)-(B2), each element $a$ has a $(d, \gamma, \tau)$-witness in $\frakB$ if it has a $(d, \gamma, \tau)$-witness in $\frakA$. Similarly for $(\bar d, \gamma, \sigma)$-witnesses.

\end{proof}

\section{Two Successors and One Linear Order}\label{section:twosuccessors}
In this section we show that $\esotwo$ is decidable on finite \mbox{$(\suc_1, \suc_2, <_2)$}-structures. This has only been known for $\emsotwo$ so far  \cite{ManuelZ13}.

\begin{theorem}\label{theorem:twosuccessors:decidability}
  $\esotwo$-satisfiability on finite ordered $(\suc_1, \suc_2, <_2)$-structures is decidable.
\end{theorem}
More precisely it is decidable as fast as emptiness for multicounter automata. 

The theorem is proved using the automata-based approach along the same lines as in Lemma~\ref{lemma:succorder}. To this end we first define \emph{linearly ordered data automata} (short: LODA), a restriction of \emph{ordered data automata} which were introduced in \cite{ManuelZ13}. Then we show that each constraint problem can be translated into a LODA such that the constraint problem has a finite $(\suc_1, \suc_2, <_2)$-solution if and only if the automaton accepts some linearly ordered data word.
Theorem~\ref{theorem:twosuccessors:decidability} then follows from Lemma \ref{lemma:formulaToConstraint} and the decidability of the emptiness problem for LODA (see Theorem \ref{theorem:loda:emptiness}).

Ordered data automata have been introduced for studying
$\emsotwo$ on $(\suc_1, \suc_2, \prec_2)$-structures where $\prec_2$ is a preorder relation. For convenience we simplify the automaton model in order to study $\esotwo$ on plain $(\suc_1, \suc_2, <_2)$-structures.

A \emph{linearly ordered data word} is a word $w = (\sigma_1, d_1) \ldots (\sigma_n, d_n)$ with $(\sigma_i, d_i)$ from $\Sigma \times \N$ such that $\{d_1, \ldots, d_n\}$ is a contiguous interval in $\N$ and $d_i \neq d_j$ for all $i \neq j$. Each linearly ordered data word represents a $(\suc_1, <_1, \suc_2, <_2)$-structures extended by unary relations interpreting symbols from a signature $T$ in a canonical way. The linear order $<_1$ and its successor are represented by the positional order, while $<_2$ and $\suc_2$ are encoded by the order of the data values $d_1, \ldots, d_n$. The unary symbols are encoded by the unary types $\Sigma$ over $T$. See Figure \ref{figure:example:PointSetDataWord} for an illustration. 

  \begin{figure}[t]
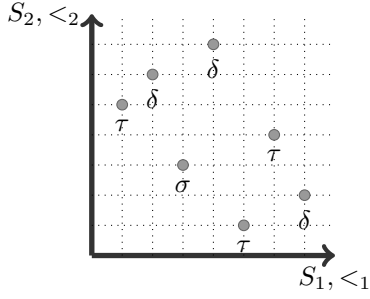
 
     \begin{center}
     \pictureExamplePointSetDataWord 
     \end{center}
    \caption{A $(\suc_1, <_1, \suc_2, <_2)$-structure represented as a point set in the two-dimensional plane. Here all other binary relations are disregarded. A linearly ordered data word corresponding to the structure is $((\tau, 5), (\delta, 6), (\sigma, 3), (\delta, 7), (\tau, 1), (\tau, 4), (\delta, 2))$.  \label{figure:example:PointSetDataWord}}
  \end{figure}

The \emph{resorting} of $w$ is the string $\sigma_{i_1} \ldots \sigma_{i_n}$ such that the data values $d_{i_1}, \ldots, d_{i_n}$ are sorted in ascending order.
The \emph{string projection} of $w$ is the string $\sigma_1 \ldots \sigma_n$. The \emph{marked string projection} of $w$ is its string projection annotated by information about the relationship of data values of adjacent positions. We make this more precise. The \emph{marking} $m_i = (m, m')$ of position $i$ is a tuple from $\Sigma_M = \{ -\infty, -1, +1, \infty, -\}^2$  and is defined as follows. If $i = 1$ (or $i = n$) then $m = -$ (or $m' = -$), otherwise: 

$$m =
                  \begin{cases}
                  -\infty & \mbox{ if } d_{i-1} < d_i - 1 \\
                  -1 & \mbox{ if }  d_{i-1} = d_i - 1  \\
                  +1 & \mbox{ if } d_{i-1} = d_i + 1 \\
                  \infty & \mbox{ if } d_{i-1} > d_i + 1
                   
                  \end{cases}
                  $$

Similarly $m'$ is defined using $d_{i+1}$ instead of $d_{i-1}$. The marked string projection of $w$ is the string $(\sigma_1, m_1)\ldots$ $(\sigma_n, m_n)$ over the alphabet~\mbox{$\Sigma \times \Sigma_M$}.The string projection of the linearly ordered data word $(\tau, 5)$ $(\delta, 6)$ $(\sigma, 3)$ $(\delta, 7)$ $(\tau, 1)$ $(\tau, 4)$ $(\delta, 2)$ from Figure \ref{figure:example:PointSetDataWord} is $\tau\delta\sigma\delta\tau\tau\delta$. Its marked string projection is $(\tau, (-, +1))$ $(\delta, (-1, -\infty))$ $(\sigma, (+\infty, +\infty))$ $(\delta, (-\infty, -\infty))$ $(\tau, (+\infty, +\infty))$ $(\tau, (-\infty, -\infty))$ \\$(\delta, (+\infty, -))$, and its resorting is $\tau\delta\sigma\tau\tau\delta\delta$.

A \emph{linearly ordered data automaton (short: LODA)} $\calA$ over alphabet $\Sigma$ is a tuple $(\calB, \calC)$ where $\calB$ is a (non-deterministic) finite state string transducer with input alphabet $\Sigma \times \Sigma_M$ (where $\Sigma_M$ is the set of markings) and some output alphabet $\Sigma'$; and $\calC$ is a finite state automaton over $\Sigma'$.

A LODA $\calA = (\calB,\calC)$ works as follows. First, for a given linearly ordered data word $w$, the transducer $\calB$ reads the marked string projection of $w$. A run $\rho_B$ of the transducer defines a unique new labeling of each position. Let $w'$ be the linearly ordered data word thus obtained from $w$. The finite state automaton $\calC$ runs over the resorting of $w'$ yielding a run $\rho_C$. The run $\rho_\calA = (\rho_B, \rho_C)$ of $\calA$ is accepting, if both $\rho_B$ and $\rho_C$ are accepting. The automaton $\calA$ accepts $w$ if there is an accepting run of $\calA$ on $w$. The set of linearly ordered data words accepted by $\calA$ is denoted by $\calL(\calA)$. We refer to~\cite{ManuelZ13} for more details.

\begin{example}\label{example:LODA}
  Consider the data language $L$ over $\Sigma = \{\sigma, \tau, \rho\}$ that contains all linearly ordered data words with a unique $\sigma$-position~$x$ with a $\tau$-position $y$ to its right such that the data value of $y$ is the successor of the data value of $x$ (see Figure \ref{figure:example:PointSetDataWord} for a linearly ordered data word in $L$). A LODA $\calA = (\calB, \calC)$ can recognize $L$ as follows. For a linearly ordered data word, the automaton $\calB$ checks that there is a unique $\sigma$-position $x$, guesses the $\tau$-position $y$ to the right of $x$ and colors this position $y$ with a fresh label $\tau'$ (using the transduction). The automaton $\calC$ checks that the $\tau'$-position is the data successor of the $\sigma$-position.
\end{example}

The following theorem follows from  Theorem 7 and Corollary 11 in~\cite{ManuelZ13}.
\begin{theorem}\label{theorem:loda:emptiness}
  Emptiness of LODA is decidable.
\end{theorem}
More precisely it is decidable as fast as emptiness for multicounter automata.

Next we show that the satisfiability problem for $\esotwo$ on finite $(\suc_1, \suc_2, <_2)$-structures reduces to the non-emptiness problem for LODAs.
\begin{lemma}\label{lemma:twosucc}
  For every constraint problem $C$ there is a LODA $\calA$ such that
  $C$ has a finite $(\suc_1, \suc_2, <_2)$-solution if and only if $L(\calA)$ is non-empty.
\end{lemma}

  \begin{proofsketch}
    Assume that $C \df (C_\exists, C_\forall)$ is a constraint problem over a signature $T \cup \{\suc_1, \suc_2, <_2\}$ and let $\Sigma \df \Sigma_{T}$ and $\Gamma \df \Gamma_{T}$. As in the proof of Lemma~\ref{lemma:succordfinitemodel} we assume that non of the possible witnesses requested by an existential constraint contradicts a universal constraint. 
    
    We construct a LODA $\calA = (\calB, \calC)$ such that $\calA$ accepts an ordered data word over $\Sigma$ if and only if $C$ has a finite $(\suc_1, \suc_2, <_2)$-solution. Intuitively the automaton $\calA$ interprets linearly ordered data words as extensions of ordered $(\suc_1, \suc_2, <_2)$-structures by unary relations but with no binary relations. In order to accept a word,  it has to verify that binary $T$-types can be assigned to all pairs of positions in a way consistent with $C$.

    The LODA $\calA$ is very similar to the finite state automaton constructed in Lemma \ref{lemma:succorder}. Like the automaton in Lemma \ref{lemma:succorder}, the LODA~$\calA$ assumes that every position is labeled with its required witnesses~(E). 
    
        \begin{enumerate}
          \item[(E)] We assume that for every existential constraint $(\sigma, E)$, all $\sigma$-labeled positions $a$ are labeled with a fresh label $(\sigma, d, \gamma, \tau)$ such that $(d, \gamma, \tau) \in E$. The intention is that the $(\sigma, E)$-witness of $a$ satisfies $(d, \gamma, \tau)$.
        \end{enumerate}
    
    Dealing with positions that are close to each other can be done in a similar same way as in Lemma \ref{lemma:succorder}, except that now \emph{close} means either close with respect to $\suc_1$ or close with respect to $\suc_2$ (or both). The existence of an assignment of binary $T$-types for positions that are close with respect to $\suc_1$ is verified by $\calB$, and it is verified by $\calC$ for positions close with respect to $\suc_2$. 
    
        More precisely, the automaton guesses the types of pairs of elements that are $\suc_1$- or $\suc_2$-close to each other.  Each element $a$ has at most four elements that are close to it: there might be an element  $b_1$ with $\suc_1(a, b_1)$, an element $b_2$ with $\suc_1(b_2, a)$ and similarly elements $b_3$  with $\suc_2(a, b_3)$ and $b_4$  with $\suc_2(b_4, a)$. Some of those elements might not exist (if $a$ is the first or last element with respect to $<_1$ or $<_2$) or might coincide (e.g. $b_1$ might coincide with either $b_3$ or $b_4$). For each of the elements $b_1, \ldots, b_4$ the automaton can guess and verify the binary $T$-type.
    
    \begin{enumerate}
      \item[(L1)] (Local types) We assume that each element $a$ is labeled by up to four labels $\gamma_1, \ldots, \gamma_4 \in \Gamma$. The intention is that $\gamma_i$ is the binary type of $(a,b_i)$ (if the element $b_i$ exists).
      \item[(L2)] (Consistency of local types) The automaton verifies that the labels are consistent, that is, e.g., that if $a$ and $b$ are elements with $\suc_1(a,b)$ then the label $\gamma_1$ of $a$ (i.e. the type guessed for $(a, b)$) is compatible with the label $\gamma_2$ of $b$ (i.e. the type guessed for $(b, a)$) and that if $\suc_1(a,b)$ and $\suc_2(a,b)$ then $\gamma_1 = \gamma_3$.
      \item[(L3)] (Local witnesses) For every $\sigma$-labeled element $a$ that is labeled by $(\sigma, d, \gamma, \tau)$ due to (E), the automaton verifies that if $d = \suc_1(x,y)$ then the label $\gamma_1$ of $a$ is $\gamma$ and that its successor is labeled with $\tau$. Likewise for $\suc_1(y,x), \suc_2(x,y)$, and $\suc_2(y,x)$).
      \item[(L4)] (Local universal constraints) For all $\sigma$- and $\tau$-labeled positions $a$ and $b$ with $\suc_1(a,b)$ the automaton verifies that if $a$ is labeled $\gamma_1$ then there is no universal constraint $(\sigma, \suc(x,y), \gamma_1, \tau)$. Likewise for $\suc_1(b,a)$,  $\suc_2(a,b)$, and  $\suc_2(b,a)$. 
    \end{enumerate}

    For verifying the existence of an assignment of binary types for positions that are remote from each other, the LODA $\calA$ has to do slightly more than the finite state automaton from Lemma \ref{lemma:succorder}. There is, essentially, only one $(\suc_1, \suc_2, <_2)$-type of remote positions which we will denote by~$d \df \neg \suc_1(x, y) \wedge \neg \suc_1(y, x) \wedge x <_2 y \wedge  \neg \suc_2(x,y)$ (see Figure \ref{figure:theorem:binary}).
    
      \begin{figure}[t]
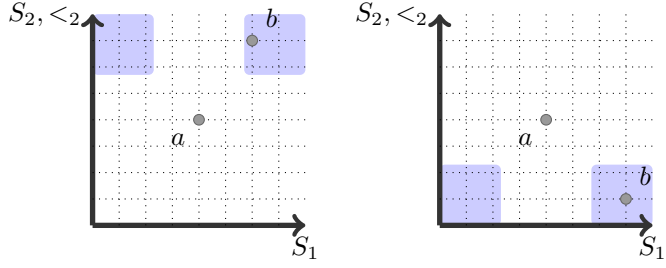
 
        \begin{center}\pictureTwoSuccessorsFarAwayA \pictureTwoSuccessorsFarAwayB \end{center}
        \caption{Illustration of remote elements in Lemma \ref{lemma:twosucc}. On the left, the element $a$ and all potential elements $b$ in the highlighted area satisfy $d(x,y) \df \neg \suc_1(x, y) \wedge \neg \suc_1(y, x) \wedge x <_2 y \wedge  \neg \suc_2(x,y)$. On the right, the symmetric constraint $\neg \suc_1(x, y) \wedge \neg \suc_1(y, x) \wedge y <_2 x \wedge  \neg \suc_2(y,x)$ is illustrated.  \label{figure:theorem:binary}}
      \end{figure}
    
    Checking that for all positions $a$ and $b$ at least one binary $T$-type is consistent with the universal constraints can be done as follows:
    \begin{enumerate}
    \item[(U1)] Assume that each position $a$ is labeled by a set $L_{<_2}(a)$ containing all $\tau \in \Sigma$ such that there is a $\tau$-labeled position $b_1$ such that $(a,b_1)$ satisfies $d$ and by a set $L_{>_2}(a)$ containing all $\tau \in \Sigma$ such that there is a $\tau$-labeled position $b_2$ such that $(b_2, a)$ satisfies $d$. Those labels can be easily verified by $\calB$ and~$\calC$. (E.g. if $\tau \in L_{<_2}(a)$ for a position $a$ and there are $i \in \{0,1,2\}$ many $\tau$-labeled $\suc_1$-close elements $b$ such that $(a,b)$ satisfies $x <_2 y \wedge \neg \suc_2(x,y)$ then $\calC$ verifies that there are $i+1$ $\tau$-labeled elements $b'$ such that $(a, b')$ satisfies $x <_2 y \wedge \neg \suc_2(x,y)$. Then $(a, b')$ satisfies $d$ for one of those $b'$. Similarly for~$\tau \notin L_{<_2}(a)$ and $L_{>_2}(a)$.)
    \item[(U2)] For all $\sigma$-labeled $a$ the automaton verifies that for every $\tau \in L_{<_2}(a)$ there is a $\gamma$ such that there is no universal constraint $(\sigma, d, \gamma, \tau)$ in $C$. Similarly for $L_{>_2}(a)$.
    \end{enumerate}

    Testing that existential witnesses can be assigned to far-away positions is analogous to Lemma \ref{lemma:succorder}. The intuition is that pairs $(a,b)$ of elements satisfying $d$ behave like pairs $(a', b')$ satisfying $x < y \wedge \neg \suc(x, y)$ in  Lemma \ref{lemma:succorder}. A little care is needed for ensuring that elements are indeed remote from each other.
    
    We shortly discuss some details of how $\calA$ verifies that binary $T$-types for all $\sigma$- and $\tau$-labeled elements $a$ and $b$ satisfying $d$ can be assigned. To this end, the automaton $\calA$ guesses whether there is an element $u$ such that $\sigma$ occurs more than $k \df 3(|\Gamma|+5)$ times above $u$ and $\tau$ occurs more than $k$ times below $u$. (Recall that a position $a$ is above $u$ if~\mbox{$a >_2 u$}.)

    If such a position $u$ exists then the automaton tests that every $\sigma$-labeled element $a$ has as many $\tau$-labeled elements $b$ such that $(a,b)$ satisfies $d$ as is required by (E), and that every $\tau$-labeled element $a'$ has sufficiently many $\sigma$-labeled positions $b'$ such that $(a', b')$ satisfies $\bar d$. If this is the case then binary types can be assigned as in Lemma \ref{lemma:succorder} using Otto's assignment technique.
  
    For testing that there are, say $m$ many, $\tau$-labeled elements $b$ such that $(a,b)$ satisfies $d$ the automaton $\calB$ labels $a$ with the number $i \in \{0,1,2\}$ of $\tau$-labeled positions that are $\suc_1$-close to $a$ and $\suc_2$-remote from~$a$. The automaton $\calC$ then tests that there are at least $m+i$ many $\tau$-labeled elements $b$ such that $(a,b)$ satisfies $x <_2 y \wedge \suc_2(x,y)$.
  
    If there is no such position $u$ then there is a position $v$ such that there are at most $k+1$ $\sigma$-labeled positions below $v$ and at most $k+1$ $\tau$-labeled positions above $v$ by Lemma \ref{lemma:richpoor}. The automaton $\calA$ exploits this structure exactly as in Lemma \ref{lemma:succorder} by labeling those up to $k+1$ many positions distinctly and guessing and verifying their witnesses.

      For completeness we describe $\calA$ in more detail, even though it is very similar to the automaton in Lemma \ref{lemma:succorder}. The automaton $\calA$ does the following: 
      \begin{enumerate}
      \item[(E1)] It guesses whether there is such a position $u$. The correctness of the guess can be easily verified by $\calC$.
      \item[(E2)] If there is such a position $u$ then $\calA$ verifies the following for all $\sigma, \tau \in \Sigma$:
        \begin{enumerate}
        \item If a $\sigma$-labeled position $a$ is labeled by $(\sigma, d, \gamma_1, \tau), \ldots, (\sigma, d, \gamma_m, \tau)$ then there are $m$ $\tau$-labeled positions above $a$ that are neither $\suc_1$- nor $\suc_2$-close to $a$.
        \item Symmetrically, if a $\tau$-labeled position $a$ is labeled by $(\tau, \bar d, \gamma_1, \sigma)$, $\ldots$, $(\tau, d, \gamma_m, \sigma)$ then there are $m$ $\sigma$-labeled positions below $a$ that are neither $\suc_1$- nor $\suc_2$-close to $a$.
        \end{enumerate}
      \item[(E3)] If there is no such position $u$ then $\calA$ guesses the position $v$. The correctness of the guess can be easily verified by $\calC$.      Then: 
        \begin{enumerate}
          \item We assume that positions are labeled by the following extra information (using fresh labels depending on $\sigma$ and $\tau$):
          \begin{itemize}
            \item The $i$th $\sigma$-labeled position $a$ below $v$ is labeled by $\sigma_i$. The $i$th $\tau$-labeled position $a$ above $v$ is labeled by $\tau_i$. (By the choice of $v$ there are at most $k+1$ such $\sigma_i$ and $\tau_i$.)
            \item The intended witnesses for $\sigma_i$- and $\tau_i$-labeled positions are labeled as follows:
              \begin{itemize}
              \item If the $\sigma_i$-labeled position $a$ is labeled by $(\sigma, \gamma, d, \tau)$ in (E), then there is a $\tau$-labeled position $b$ labeled with $(\sigma_i, \gamma, d, \tau)$ such that $(a,b)$ satisfies $d$. 
              \item Likewise, if the $\tau_i$-labeled position $a$ is labeled by $(\tau, \gamma, \bar d, \sigma)$ in (E), then there is a $\sigma$-labeled position $b$ that is labeled with $(\tau_i, \gamma, \bar d, \sigma)$ such that $(a, b)$ satisfies $\bar d$ . 
              \end{itemize}
            \item Positions with intended witnesses from $\sigma_i$- and $\tau_i$-labeled positions are labeled as follows:
              \begin{itemize}
                \item  Each position $a$ below $v$ that is labeled by $(\tau, \gamma, \bar d, \sigma)$ in (E) is also labeled with $(\tau, \gamma, \bar d, \sigma_i)$ for some $i$ such that $(a,b)$ satisfies $\bar d(x,y)$ where $b$ is the $\sigma_i$-labeled position.
                \item  Likewise, each $(\sigma, \gamma, d, \tau)$-labeled position $a$ above $u$ is labeled with $(\sigma, \gamma, d, \tau_i)$ for some $i$ such that $(a, b)$ satisfies $d$ where $b$ is the $\sigma_i$-labeled position. 
              \end{itemize}
          \end{itemize}
          \item The automaton verifies that the labels are consistent, that is:
            \begin{itemize}
              \item No $\tau$-labeled position below $v$ is labeled with $(\sigma_i, \gamma, d, \tau)$ and with $(\tau, \gamma', \bar d, \sigma_i)$ where $\gamma$ and $\gamma'$ are not reverse types.
              \item Likewise, no $\sigma$-labeled position above $v$ is labeled with $(\tau_i, \bar \gamma', d, \sigma)$ and with $(\sigma, \gamma, d, \tau_i)$ where $\gamma$ and $\gamma'$ are not reverse types.
            \end{itemize}
        \end{enumerate}
    \end{enumerate}

    The correctness is proved completely analogous to Lemma \ref{lemma:succorder}.
                                                                                                                                                                                                              \end{proofsketch}

\section{Order- and Successor-invariance}\label{section:invariance}
In this Section we discuss order-invariance for two-variable formulas. A first-order sentence $\varphi$ over a signature $T \cup \{<\}$ is \emph{order-invariant} (\emph{$<$-invariant} for short) if for each $T$-structure $\frakA$ and all linear orders $<_{1}$ and $<_{2}$ on the domain of $\calA$,
$$ (\calA,<_{1}) \models \varphi \ \iff \ (\calA, <_{2}) \models \varphi. $$
A class of finite $T$-structures $\calA$ is \emph{$<$-invariantly first-order definable} if there is an $<$-invariant $T\cup\{<\}$-sentence $\varphi$ such that, for each finite $T$-structure $\frakA$, $$\text{$\frakA \in \calC$ if and only $(\frakA,<) \models \varphi$}$$ for each linear order $<$ on the domain of $\frakA$.

It is not immediately obvious that allowing this restricted use of an order extends the expressive power of $\fo$.
A well-known example due to Gurevich (cf. eg. \cite{Libkin2004}) shows that there is indeed a $<$-invariantly definable class of structures which is not $\fo$-definable without an order. The example of Gurevich indicates that order-invariance is, at least potentially, useful in formulating queries. For using invariance in this context, it is essential to be able to verify that an $\fo$-sentence is indeed invariant. Unfortunately, a simple reduction of the finite satisfiability problem shows that $<$-invariance of $\fo$-sentences is undecidable.

The results from the previous sections as well as the discussion in the introduction imply that $<$-invariance of $\fotwo$-sentences is decidable.
\begin{theorem}\label{thm:inv-decidable}
  Order-invariance of $\fotwo$ is in $\TWONEXPTIME$.
\end{theorem}

\emph{Successor-invariance} (\emph{$\suc$-invariance}, for short) is defined analogously to order-invariance where instead of the linear order the formulas may use a successor relation (see e.g. \cite{Rossman2003}). By combining the approach discussed in the introduction with Theorem 2 of \cite{CharatonikW13}, one obtains the following result.
\begin{theorem}
  Successor-invariance of $\fotwo$ is in $\NEXPTIME$.
\end{theorem}

Our approach for deciding $<$-invariance and $\suc$-invariance does not immediately transfer to $(\suc, <)$-invariance
where both the order and its induced successor can be used in formulas, since even $\emsotwo$ is undecidable on finite $(\suc_1, <_1, \suc_2, <_2)$-structures \cite{Manuel10}
However, in the proof of Theorem~\ref{thm:inv-decidable}, it would suffice to prove decidability for formulas of the form $\varphi(\suc_{1}, <_{1}) \wedge \lnot\varphi(\suc_{2}, <_{2})$.
This seems conceivable, but so far we were unable to obtain a proof using our techniques.

The following example shows that invariance is useful even for two-variable logic.
\begin{example}
$\fotwo$ is often extended by \emph{counting quantifiers} of the shape $\exists^{\geq k}$, where $\exists^{\geq k} x\, \varphi$ states that there are at least $k$ satisfying assignments to $x$. We define \emph{monadic counting quantifiers} which are defined in the same way but $\varphi$ is restricted to use at most one free variable $x$. 
We observe that these quantifiers are $<$-invariantly definable by two-variable formulas, since
\begin{align*}
  \exists^{\geq 0} x\, \varphi(x) \ &\equiv \ x= x,\\
  \exists^{\geq k+1} x\, \varphi(x) \ &\equiv \  \exists x\, \varphi(x) \wedge \exists^{\geq k} y \ y < x \wedge \varphi(y).
\end{align*}
It is easy to see using a pebble game argument that the monadic counting quantifier $\exists^{\leq k}$ is not two-variable definable for each~\mbox{$k \geq 3$}.
\end{example}

\section{Conclusion and Future Work}\label{section:conclusion}
We have shown that $<$-invariance and $\suc$-invariance of $\fotwo$-sentences is decidable by establishing decidability of $\esotwo$-satisfiability on finite $(\suc_1, <_1, <_2)$- and $(\suc_1, \suc_2, <_2)$-structures. Several interesting questions remain open:

\begin{enumerate}
 \item Is $\esotwo$ decidable on $(\suc_1, \suc_2, \suc_3)$-structures? 
 \item Where is the border of decidability of $\esotwo$ on general, not necessarily finite, ordered structures?
 \item Is $(\suc, <)$-invariance of $\fotwo$ decidable?
 \item Is every $<$-invariantly $\fotwo$-definable property also $\fo$-definable without the use of a linear order?
\end{enumerate}
For the second question preliminary results have been obtained for certain order types.

\section*{Acknowledgements}

We thank Thomas Schwentick for stimulating discussions and many very helpful suggestions for improving upon a draft of this article. The first author acknowledges the financial support by DFG grant \mbox{SCHW 678/6-1}.

\bibliography{bibliography}

\end{document}